\documentclass[pdflatex,sn-mathphys-num]{sn-jnl}% Math and Physical Sciences Numbered Reference Style
\usepackage{graphicx,subcaption,epstopdf}%
\usepackage{multirow}%
\usepackage{amsmath,amssymb,amsfonts}%
\usepackage{amsthm}%
\usepackage{mathrsfs}%
\usepackage[title]{appendix}%
\usepackage{xcolor}%
\usepackage{textcomp}%
\usepackage{manyfoot}%
\usepackage{booktabs}%
\usepackage{algorithm}%
\usepackage{algorithmicx}%
\usepackage{algpseudocode}%
\usepackage{listings}%
\usepackage{bm}%
\usepackage{booktabs,array}
\theoremstyle{thmstyleone}%
\newtheorem{theorem}{Condition}%  meant for continuous numbers
\newtheorem{proposition}{Proposition}% to get separate numbers for theorem and proposition etc.

\theoremstyle{thmstyletwo}%

\theoremstyle{thmstylethree}%

\begin{document}

\title[Shifted asymmetric Laplace mixtures of experts]{Shifted asymmetric Laplace mixtures of experts}

\author*[1]{\fnm{Sphiwe B.} \sur{Skhosana}}

\author[2,3]{\fnm{Hien Duy} \sur{Nguyen}}

\affil[1]{\orgdiv{Department of Statistics}, \orgname{University of Pretoria}, \orgaddress{\city{Pretoria}, \postcode{0002}, \country{South Africa}}}
\affil[2]{\orgdiv{School of Computing, Engineering and Mathematical Sciences}, \orgname{La Trobe University}, \orgaddress{\city{Bundoora}, \postcode{VIC 3086}, \country{Australia}}}
\affil[3]{\orgdiv{Institute of Mathematics for Industry}, \orgname{Kyushu University}, \orgaddress{\city{Nishi Ward}, \postcode{Fukuoka 819-0395}, \country{Japan}}}

\abstract{Mixtures of experts (MoE) models provide a flexible framework for modelling heterogeneity in data for regression and model-based clustering and classification. MoE models for regression are typically based on the Gaussian assumption for the expert distributions. To robustify the MoE framework with respect to data exhibiting skewness, heavy tails and outliers, we propose a robust non-normal MoE model using the shifted asymmetric Laplace (SAL) distribution. The proposed SALMoE model overcomes the limitations of the Gaussian MoE model when the observed data are asymmetric and heavy-tailed. Through a combination of the minorization-maximization (MM) algorithm with the classical Expectation-Maximization (EM), we develop a dedicated hybrid EM-MM algorithm to estimate the parameters of the SALMoE model. The EM-MM algorithm is shown to yield a nondecreasing observed log-likelihood. A simulation study demonstrates the robustness and practical utility of the proposed model. Finally, the SALMoE model is applied to two real-world economic datasets. The \texttt{R} code used to fit the SALMoE model using the EM-MM algorithm is available through the link: \href{https://github.com/Sphiwe-Skhosana/SALMoE}{https://github.com/SALMoE\_EM-MM}.}




\maketitle

\section{Introduction}\label{sec1}
Finite mixtures of regression (FMR) models, first introduced by \citet{quandt1972}, are based on the assumption that the data, which consist of the response variable $y\in \mathbb{R}$ and a covariate vector $\mathbf{x}=(1,x_1,\dots,x_p)\in\mathbb{R}^{p+1}$ including an intercept, are generated from, say $K$, unknown regression models. Moreover, the data are governed by a component indicator variable $Z$ indicating which regression model generated which data point. The variable $Z$ has distribution $\mathbb{P}(Z=k)=\pi_k$, for $k=1,2,\dots,K$, where the $\pi_k$'s are the non-negative mixing probabilities satisfying $\sum_{k=1}^K\pi_k=1$. The variable $Z$ is typically unobserved or latent. Thus, given the set of covariates $\mathbf{x}$, the conditional density of the response variable $y$ is defined as 
\begin{eqnarray}\label{model1}
f(y|\mathbf{x};\boldsymbol{\vartheta})=\sum_{k=1}^K\pi_kf_k(y|\mathbf{x};\boldsymbol{\theta}_k)
\end{eqnarray}
where $\boldsymbol{\vartheta}=(\pi_1,\pi_2,\dots,\pi_K;\boldsymbol{\theta}_1,\boldsymbol{\theta}_2,\dots,\boldsymbol{\theta}_K)$ is the vector of model parameters, with $\boldsymbol\theta_k$ being the parameter vector of the $k^{th}$ component conditional distribution.\\
Although flexible compared to a single regression model, model \eqref{model1} is still limited by the assumption that the mixing probabilities are constants independent of any covariate information. Mixtures of Experts (MoE) models, first introduced by \citet{jacobs1991}, increase the flexibility of FMRs by allowing the mixing probabilities to be functions of a set of $q$ covariates, $\mathbf{t}=(1,t_1,\dots,t_q)\in\mathbb{R}^{q+1}$, such that
\begin{eqnarray}\label{model2}
f(y|\mathbf{x},\mathbf{t};\boldsymbol{\vartheta})=\sum_{k=1}^K\pi_k(\mathbf{t}|\boldsymbol\eta)f_k(y|\mathbf{x};\boldsymbol{\theta}_k)
\end{eqnarray}
In the MoE model \eqref{model2}, the mixing probabilities $\pi_k(\mathbf{t}|\boldsymbol\eta)$, known as gating functions, are typically modelled using the multinomial logistic model as
\begin{eqnarray}\label{logistic}
\pi_k(\mathbf{t}|\boldsymbol\eta)=\mathbb{P}(Z=k|\mathbf{t})=\frac{\exp(\mathbf{t}^\top\boldsymbol{\eta}_k)}{\sum_{\ell=1}^K\exp(\mathbf{t}^\top\boldsymbol{\eta}_\ell)}
\end{eqnarray}
where $\boldsymbol{\eta}=(\boldsymbol{\eta}_1,\boldsymbol{\eta}_2,\dots,\boldsymbol{\eta}_K)$, with $\boldsymbol{\eta}_k=(\eta_{k0},\eta_{k1},\dots,\eta_{kq})$, is a vector of the parameters for the gating functions and $\boldsymbol{\eta}_K=\mathbf{0}$ for identifiability. The vector of the model parameters is $\boldsymbol{\vartheta}=(\boldsymbol{\eta};\boldsymbol{\theta}_1,\boldsymbol{\theta}_2,\dots,\boldsymbol{\theta}_K)$.\\
In the MoE model \eqref{model2}, the conditional densities $f_k(y|\mathbf{x};\boldsymbol{\theta}_k)$ are known as experts because different component regression models can model data in different regions of the input space. Thus, each component regression model is an expert at making predictions in its own region and the gating function assigns each component to the region where it has the most predictive power.\\
Due to their flexibility, MoE models have been used in a wide variety of applications including, among others, the analysis of rank data \cite{gormley2008}, network data \cite{gormley2010}, time-series data \cite{zeevi1996, carvalho2005,fruhwirth2012,niu2024}, heavy-tailed data \cite{nguyen2016,chamroukhi2016,mirfarah2024}, censored data \cite{mirfarah2021}, asymmetric data \cite{chamroukhi2017} and functional data \cite{tamo2024,chamroukhi2024}. Detailed comprehensive reviews of the theory and application of MoE models are given in \cite{yuksel2012}, \cite{masoudnia2014} and recently \cite{nguyen2018}.\\
For simplicity and mathematical tractability, MoE models are typically based on Gaussian experts. That is, each conditional density function $f_k(y|\mathbf{x};\boldsymbol{\theta}_k)$ is assumed to be Gaussian. However, in practice, data exhibit asymmetric behaviour, heavy tails and are often contaminated with outliers or atypical observations. In such cases, the use of the Gaussian distribution is inappropriate and may have a negative impact on the fit of the MoE model. This is because the Gaussian distribution is symmetric and sensitive to outliers.\\
The inadequacy of the Gaussian distribution for MoE modelling when the data are asymmetric, heavy-tailed and/or contaminated by outliers has been previously considered by many authors. For symmetric heavy-tailed data in the presence of outliers, \citet{nguyen2016} proposed Laplace MoE (LMoE) models in which each conditional distribution is assumed to follow a Laplace distribution; \citet{chamroukhi2016} proposed TMoE models in which each conditional distribution is assumed to follow a Student-t distribution; \citet{mirfarah2021} proposed scale mixture of Gaussian MoE (SMGMoE) models in which each conditional distribution is assumed to follow a scale mixture of Gaussian distributions. The SMGMoE model includes the LMoE and TMoE as special cases. Another useful special case of the SMGMoE model is the contaminated Gaussian MoE model in which each conditional distribution is assumed to follow a contaminated Gaussian distribution. See \cite{mambo2026} for a formal definition of the model and its non-parametric extension. The case of fitting MoE models for asymmetric heavy-tailed data with outliers has received limited attention. To the best of our knowledge, the only paper that investigates this case is \citet{chamroukhi2017}, who proposed skew-t MoE (STMoE) models in which each conditional distribution follows the skew-t distribution. However, the STMoE model is less tractable: its Expectation-step in the EM algorithm is intractable and estimation is computationally expensive.\\
In this paper, we propose shifted asymmetric Laplace MoE (SALMoE) models in which each conditional distribution is assumed to follow the shifted asymmetric Laplace (SAL) distribution (see Section \ref{subsubsec2_2} for more details). The SALMoE model provides a mathematically tractable and computationally straightforward MoE model. \citet{franczak2014} introduced the finite mixture of SAL distributions for model-based clustering and classification of non-Gaussian (asymmetric and heavy-tailed with outliers) data. The authors demonstrated that this approach performs favourably compared to the finite mixture of Gaussian distributions. \citet{sun2019} proposed to use the finite mixture of SAL distributions for image segmentation in order to achieve robustness against noise and outliers in the images. \citet{morris2019} proposed finite mixtures of contaminated SAL distributions which extend the model proposed in \cite{franczak2014} to allow for the automatic detection of outliers and consequently the quantification of the proportion of outliers present in the data. Note, however, that the models proposed in \cite{franczak2014} and \citet{morris2019} assume that both the mixing probabilities and component distributions are independent of any covariate information. Recently, \citet{otto2025} proposed finite mixtures of linear SAL regression models which extend the model proposed by \citet{morris2019} to the regression context. However, the authors assumed that the mixing proportions are independent of any covariate information. The proposed SALMoE model considers the case in which both the mixing proportions and the component distributions are functions of a set of covariates.
Compared to the skew-t expert distributions used in \cite{chamroukhi2017}, the SAL expert distributions have the advantage of being mathematically tractable and computationally straightforward as shown in \cite{franczak2014}.\\
Inference for MoE models can be conducted under the frequentist framework using maximum likelihood estimation (MLE). \citet{jordan1994} derived the MLEs of the MoE model using the Expectation-Maximization (EM) algorithm \cite{dempster1977}. \citet{ng2004} proposed to use the Expectation-Conditional-Maximization (ECM) algorithm \cite{meng1993} to obtain the parameter estimates of the MoE model. \citet{gormley2008} used a hybrid between the EM algorithm and the Minorization-Maximization (MM) algorithm \cite{lange2000} to derive the MLEs. \citet{nguyen2016} used an MM algorithm to derive the MLEs for their proposed LMoE model.\\
Inference for MoE models under the Bayesian framework is performed using Markov chain Monte Carlo (MCMC) methods (see \cite{peng1996}) and variational inference methods (see \cite{Bishop2003}).\\
In this paper, we make use of frequentist methods. In particular, we will make use of a hybrid EM-MM algorithm to derive the parameter estimates of the proposed SALMoE model. The EM-MM algorithm replaces the computationally intensive iterative reweighted least squares (IRLS) step in the traditional EM algorithm for estimating the gating parameters $\boldsymbol\eta$ with a one-step MM update similar to \cite{gormley2008}. We show that the one-step MM update ensures a nondecreasing observed log-likelihood at each iteration of the EM-MM algorithm. \\
The rest of the paper is structured as follows. In Section \ref{subsec2}, we define the proposed SALMoE model. In Section \ref{subsec3}, we derive the MLEs via the EM-MM algorithm. In Section \ref{subsec4}, we discuss model selection for the proposed SALMoE model. In Sections \ref{subsec5} and \ref{subsec6}, we show how the model can be used for prediction and model-based clustering. In Section \ref{sec3}, we demonstrate the effectiveness of the proposed methods through an extensive simulation study. In Section \ref{sec4}, we demonstrate the practical usefulness of the proposed methods on real data and then we conclude the paper and provide directions for future studies in Section \ref{sec5}.
\section{Methodology}\label{sec2}

\subsection{Notation and standing conventions}
All random variables are taken to be defined on a common probability space $(\Omega,\mathcal{F},\mathbb{P})$. Densities are with respect to Lebesgue measure, except for the latent component indicators, which are discrete. We write $\mathbb{P}_{\boldsymbol\vartheta}$, $\mathbb{E}_{\boldsymbol\vartheta}$ and $\mathbb{V}_{\boldsymbol\vartheta}$ when the dependence on the model parameter $\boldsymbol\vartheta$ is being emphasised, and we write $\mathbb{R}^{+}=(0,\infty)$. The covariate vectors $\mathbf{x}=(1,x_1,\dots,x_p)\in\mathbb{R}^{p+1}$ and $\mathbf{t}=(1,t_1,\dots,t_q)\in\mathbb{R}^{q+1}$ include intercept terms, and the corresponding supports are denoted by $\mathcal{X}$ and $\mathcal{T}$. For matrix notation, $\mathbf{A}^{\top}$ denotes transpose, $\mathbf{I}_{r}$ is the $r\times r$ identity matrix, $\mathbf{1}$ is a vector of ones of the required length, $\operatorname{diag}(\cdot)$ denotes a diagonal matrix, and $\otimes$ denotes the Kronecker product. The symbols $\preceq$ and $\succeq$ denote the Loewner order. The full parameter vector of the SALMoE model is $\boldsymbol\vartheta=(\boldsymbol\eta;\boldsymbol\theta_1,\dots,\boldsymbol\theta_K)$, where $\boldsymbol\theta_k=(\alpha_k,\sigma_k,\boldsymbol\beta_k)$, $\sigma_k>0$, and the multinomial logit gating parameters satisfy the baseline constraint $\boldsymbol\eta_K=\mathbf{0}$. Thus the natural parameter space is a subset of
\[
\Theta=\Theta_{\boldsymbol\eta}\times\prod_{k=1}^{K}
\left(\mathbb{R}\times\mathbb{R}^{+}\times\mathbb{R}^{p+1}\right).
\]
When compactness is needed for the convergence and information criteria statements, we work on the compact subset
\[
\begin{aligned}
\Theta_c=\{\boldsymbol\vartheta\in\Theta:\;&\boldsymbol\eta_K=\mathbf{0},\;\sigma_{\min}\le\sigma_k\le\sigma_{\max},\;|\alpha_k|\le A,\;\|\boldsymbol\beta_k\|\le B,\;\|\boldsymbol\eta_k\|\le E,\;\\
&k=1,\dots,K\}.
\end{aligned}
\]
for fixed constants $0<\sigma_{\min}<\sigma_{\max}<\infty$ and $A,B,E<\infty$, where $\|\cdot\|$ denotes the Euclidean norm. When identifiability is invoked, the ordering and irreducibility restrictions are imposed separately.

\subsection{The shifted asymmetric Laplace mixture of experts}\label{subsec2}
In this section, we define the proposed SALMoE model. The proposed model is based on the stochastic representation of the SAL distribution. The latter relies on the generalised inverse Gaussian (GIG) distribution. We first recall the GIG distribution followed by the SAL distribution and its stochastic representation. Finally, we describe the hierarchical representation of the SAL distribution and use it to derive the proposed SALMoE model. 
\subsubsection{The generalised inverse Gaussian distribution}\label{subsubsec2_1}
A random variable $W\in \mathbb{R}^{+}$ is said to follow a GIG distribution if it has density
\begin{equation}\label{gig_density}
    h(w)=\frac{(a/b)^{\nu/2}}{2K_{\nu}(\sqrt{ab})}w^{\nu-1}\exp\left\{-\frac{1}{2}\left(bw^{-1}+aw\right)\right\}
\end{equation}
where $a,b\in\mathbb{R}^{+}$ and $\nu\in\mathbb{R}$ are the parameters. Here $K_\nu(\cdot)$ is the modified Bessel function of the third kind with index $\nu$. One useful representation of $K_\nu(z)$, for $z>0$, is given by
\[
K_\nu(z)=\frac{1}{2}\int_{0}^{\infty}s^{\nu-1}\exp\left\{-\frac{z}{2}\left(s+s^{-1}\right)\right\}ds.
\]

\subsubsection{The shifted asymmetric Laplace distribution}\label{subsubsec2_2} 
The SAL distribution has demonstrated excellent performance in modelling asymmetric data \cite{franczak2014}. Moreover, due to its relationship to the GIG distribution, the SAL distribution is mathematically tractable and parametric estimation is computationally straightforward. The probability density function of a random variable $y\in\mathbb{R}$ that follows the SAL distribution is
\begin{equation}\label{sal_density}
    g(y|\alpha,\sigma,\mu)=\frac{2\exp\{\alpha(y-\mu)/\sigma\}}{\sqrt{2\pi\sigma}}\left[\frac{(y-\mu)^2/\sigma}{2+\alpha^2/\sigma}\right]^{\nu/2}K_\nu(u)
\end{equation}
where $\alpha\in \mathbb{R}$ is the skewness parameter, $\sigma\in\mathbb{R}^{+}$ is the scale parameter, $\mu\in\mathbb{R}$ is the location (shift) parameter, $\nu=1/2$ and
\begin{equation}\nonumber
u=\sqrt{(2+\alpha^2/\sigma)(y-\mu)^2/\sigma}
\end{equation}
According to \citet{kotz2001}, the SAL distribution can be represented as follows. Let $Z\sim\mathcal{N}(0,\sigma)$, and let $V\in \mathbb{R}^{+}$ be a random variable, independent of $Z$, that follows an exponential distribution with mean $1$, that is $V\sim \text{Exp}(1)$, then 
\begin{equation}\label{stochastic_rep}
    Y=\mu+\alpha V+Z\sqrt{V}
\end{equation}
follows the SAL distribution with density function \eqref{sal_density}. By the conditional-density identity, the conditional distribution of $V$ given $Y=y$ is 
\begin{equation}\label{bayes1}
    f(v|Y=y)=\frac{f_{Y|V}(y|v)f_V(v)}{f_Y(y)}
\end{equation}
where $Y|V=v\sim \mathcal{N}\{\mu+\alpha v,v\sigma\}$ and $V\sim \text{Exp}(1)$. As given in \cite{franczak2014} for the multivariate case, the expression for \eqref{bayes1} in the univariate case is as follows
\begin{equation}\label{bayes2}
f(v|Y=y)=\frac{v^{\nu-1}}{2}\left[\frac{(y-\mu)^2/\sigma}{2+\alpha^2/\sigma}\right]^{-\nu/2}\frac{\exp\{-\frac{1}{2v}(y-\mu)^2/\sigma-\frac{v}{2}(2+\alpha^2/\sigma)\}}{K_{\nu}(u)}
\end{equation}
With $a=2+\alpha^2/\sigma$ and $b=(y-\mu)^2/\sigma$ in \eqref{bayes2}, the resulting density function can then be seen to be a GIG density function in \eqref{gig_density}. This result will be very useful in the implementation of the estimation procedure for the proposed SALMoE model. 
\subsubsection{The proposed model}\label{subsubsec2_3} 
The SALMoE model is a regression mixture-of-experts extension of finite mixtures of SAL distributions first introduced by \citet{franczak2014} for model-based classification and clustering of multivariate non-Gaussian data. For univariate non-Gaussian data, the probability density function of a $K$ component mixture of SAL distributions is
\begin{equation}
f(y|\boldsymbol{\vartheta})=\sum_{k=1}^K\pi_kg(y|\alpha_k,\sigma_k,\mu_k)
\end{equation}
where $\boldsymbol{\vartheta}=((\pi_1,\boldsymbol{\theta}_1),(\pi_2,\boldsymbol{\theta}_2),\dots,(\pi_K,\boldsymbol{\theta}_K))$, with $\boldsymbol{\theta}_k=(\alpha_k,\sigma_k,\mu_k)$, is a vector of all the model parameters and the $k^{th}$ component density is given by \eqref{sal_density}.\\
\citet{morris2019} proposed mixtures of contaminated SAL distributions for model-based clustering of asymmetric data in the presence of outliers. However, the authors assumed that the mixing proportions and the component distributions are independent of any covariate information. \citet{otto2025} extended the model proposed by \cite{morris2019} to the regression case by assuming that each component distribution is a function of a set of random covariates. However, the model assumes that the mixing proportions are constant and hence not influenced by any covariate information.\\
The proposed SALMoE is an MoE model with SAL distributed experts and gating functions obtained by modelling the mixing probabilities as functions of the covariates using the multinomial logistic model in \eqref{logistic}. For the covariates $\mathbf{x}$ and $\mathbf{t}$ defined in the standing conventions, the conditional density function of $y$ is 
\begin{equation}\label{SALMoE}
f(y|\mathbf{x},\mathbf{t};\boldsymbol{\vartheta})=\sum_{k=1}^{K}\pi_k(\mathbf{t}|\boldsymbol\eta)g(y|\alpha_k,\sigma_k,\mu_{k}(\mathbf{x};\boldsymbol{\beta}_k))
\end{equation}
having parameter vector $\boldsymbol{\vartheta}=(\boldsymbol\eta;\boldsymbol\theta_1,\boldsymbol\theta_2,\dots,\boldsymbol\theta_K)$, where $\boldsymbol\theta_k=(\alpha_k,\sigma_k,\boldsymbol\beta_k)$ is the vector of model parameters for the $k^{th}$ expert and $\mu_k(\mathbf{x};\boldsymbol\beta_k)=\mathbf{x}^\top\boldsymbol\beta_k$, with $\boldsymbol{\beta}_k=(\beta_{k0},\beta_{k1},\dots,\beta_{kp})$ the vector of the regression parameters for the $k^{th}$ expert component.

\subsubsection{Identifiability of the SALMoE model}
The identifiability of an MoE model or mixture model, in general, is important in order to conduct meaningful and accurate inference. \citet{titterington1985} and \citet{hennig2000} established the identifiability (up to label switching) of mixture models and mixtures of regressions, respectively.\\
\citet{jiang1999} established the identifiability of ordered, initialized and irreducible MoE models. In the SALMoE model, initialized means that the gating functions have the multinomial logistic form in \eqref{logistic} with the baseline constraint $\boldsymbol\eta_K=\mathbf{0}$; ordered means that an ordering relation is imposed on the expert parameters $\boldsymbol\theta_k$ such that $\boldsymbol\theta_1\prec\boldsymbol\theta_2\prec\dots\prec\boldsymbol\theta_K$; and irreducible means that if $k\neq j$, then at least one of $\alpha_k\neq\alpha_j$, $\sigma_k\neq\sigma_j$ or $\boldsymbol\beta_k\neq\boldsymbol\beta_j$ holds. The following non-degeneracy condition is the SAL analogue of the condition used in Lemma 2 of \citet{jiang1999}.
\begin{theorem}[Non-degeneracy condition]\label{cond:nondegeneracy}
For any $2K$ distinct triplets $(\alpha_r,\sigma_r,\boldsymbol\beta_r)$, $r=1,2,\dots,2K$, with $\sigma_r>0$, the functions
\[
(y,\mathbf{x})\mapsto g(y|\alpha_r,\sigma_r,\mu_r(\mathbf{x};\boldsymbol\beta_r))
\]
are linearly independent on $\mathbb{R}\times\mathcal{X}$.
\end{theorem}
\begin{proposition}\label{prop:identifiability}
For a fixed number of components $K$, any irreducible, ordered and initialized SALMoE model satisfying Condition \ref{cond:nondegeneracy} is identifiable up to label switching. Under the imposed ordering, the model is identifiable.
\end{proposition}
\begin{proof}
Consider two $K$-component SALMoE representations that give the same conditional density for all $(y,\mathbf{x},\mathbf{t})\in\mathbb{R}\times\mathcal{X}\times\mathcal{T}$. The multinomial logistic gating functions with the baseline constraint $\boldsymbol\eta_K=\mathbf{0}$ satisfy the initialized gating condition in \citet{jiang1999}. Irreducibility excludes duplicate expert components, and Condition \ref{cond:nondegeneracy} gives the required non-degeneracy of the expert family. Lemma 2 of \citet{jiang1999} therefore implies that the two representations have the same gating and expert parameters after a common permutation of the component labels. The ordering constraint removes this remaining permutation ambiguity, giving identifiability under the imposed ordering.
\end{proof}
\subsection{Estimation procedure for the SALMoE model}\label{subsec3}
Consider a random sample $\{(\mathbf{x}_i,\mathbf{t}_i,y_i):i=1,2,\dots,n\}$ generated from the model \eqref{SALMoE}. The corresponding observed-data log-likelihood is
\begin{equation}\label{loglik}
\ell(\boldsymbol\vartheta)=\sum_{i=1}^{n}\log\left\{\sum_{k=1}^{K}\pi_k(\mathbf{t}_i|\boldsymbol\eta)g(y_i|\alpha_k,\sigma_k,\mu_{k}(\mathbf{x}_i;\boldsymbol{\beta}_k))\right\}.
\end{equation}
The estimation of the parameter $\boldsymbol\vartheta$ is performed by making use of a hybrid algorithm that combines the Minorization-Maximization (MM) algorithm \cite{lange2000} with the classical Expectation-Maximization (EM) algorithm \cite{dempster1977}, which we call the EM-MM algorithm. Towards that end, for each data point $(\mathbf{x}_i,\mathbf{t}_i,y_i)$, we introduce the unobserved component indicator variable $\mathbf{z}_i=(z_{i1},z_{i2},\dots,z_{iK})$, where $z_{ik}=1$ if the data point $(\mathbf{x}_i,\mathbf{t}_i,y_i)$ belongs to the $k^{th}$ component or $0$ otherwise and the latent variable $v_i$ which results in the complete-data $\{(\mathbf{x}_i,\mathbf{t}_i,\mathbf{z}_i,v_i,y_i):i=1,2,\dots,n\}$.\\
Now, we can obtain the hierarchical representation for the SALMoE model using the stochastic representation \eqref{stochastic_rep} as
\begin{eqnarray}
    Y_i|\mathbf{x}_i,v_{i},z_{ik}=1&\sim& \mathcal{N}\{\mathbf{x}_i^\top\boldsymbol\beta_k+v_{i}\alpha_k,v_{i}\sigma_k\}\nonumber\\
    v_{i}|z_{ik}=1&\sim& \text{Exp}(1)\nonumber\\
    \mathbf{z}_i|\mathbf{t}_i&\sim& \text{Mult}(1;\pi_1(\mathbf{t}_i|\boldsymbol\eta),\dots,\pi_K(\mathbf{t}_i|\boldsymbol\eta))\nonumber
\end{eqnarray}
where $\text{Mult}(1;\phi_1,\dots,\phi_m)$ denotes the multinomial distribution with a single trial, $m$ categories each having success probabilities $\phi_1,\phi_2,\dots,\phi_m$.\\
Based on the above, the corresponding complete-data log-likelihood is 
\begin{equation}
    \ell_c(\boldsymbol\vartheta)=\ell_{c}(\boldsymbol\eta)+\ell_{c}(\boldsymbol\theta),
\end{equation}
where
\begin{align}
\ell_{c}(\boldsymbol\eta)&=\sum_{i=1}^n\sum_{k=1}^K z_{ik}\log\{\pi_k(\mathbf{t}_i|\boldsymbol\eta)\},\\
\ell_{c}(\boldsymbol\theta)&=\sum_{i=1}^n\sum_{k=1}^K z_{ik}\log\{f(y_i|\mathbf{x}_i,v_i,z_{ik}=1)\times f(v_i|z_{ik}=1)\}\nonumber\\
&=\sum_{i=1}^n\sum_{k=1}^K z_{ik}\log\left[\mathcal{N}\{y_i|\mu_k(\mathbf{x}_i;\boldsymbol\beta_k)+v_i\alpha_k,v_i\sigma_k\}\times \exp(-v_i)\right]\nonumber\\
&=\sum_{i=1}^n\sum_{k=1}^K z_{ik}\left[-\frac{1}{2}\log(2\pi\sigma_kv_i)-\frac{\{y_i-\mu_k(\mathbf{x}_i;\boldsymbol\beta_k)-v_i\alpha_k\}^{2}}{2\sigma_kv_i}-v_i\right]\nonumber\\
&=\sum_{i=1}^n\sum_{k=1}^K z_{ik}\left[
\begin{aligned}
&-\frac{1}{2}\log(2\pi\sigma_kv_i)-\frac{\{y_i-\mu_k(\mathbf{x}_i;\boldsymbol\beta_k)\}^{2}}{2\sigma_kv_i}\\
&+\frac{\alpha_k\{y_i-\mu_k(\mathbf{x}_i;\boldsymbol\beta_k)\}}{\sigma_k}-\frac{v_i\alpha^2_k}{2\sigma_k}-v_i
\end{aligned}
\right].
\end{align}
The EM-MM algorithm begins with an appropriately chosen initial parameter vector $\boldsymbol\vartheta^{(0)}$ and then alternates between the Expectation (E-) step, Maximization (M-) step and MM-step until convergence. In the E-step of the EM algorithm, we calculate the conditional expected value of the complete-data log-likelihood $\ell_c(\boldsymbol\vartheta)$, denoted by $Q(\boldsymbol\vartheta|\boldsymbol\vartheta^{(old)})$, given the observed data $\{(\mathbf{x}_i,\mathbf{t}_i,y_i):i=1,2,\dots,n\}$ and parameter estimate $\boldsymbol\vartheta^{(old)}$. In the M-step, we maximize $Q(\boldsymbol\vartheta|\boldsymbol\vartheta^{(old)})$ with respect to $\boldsymbol\vartheta$. 
\subsubsection*{E-step}
In the E-step of the EM algorithm, we obtain the conditional expectation of $\ell_c(\boldsymbol\vartheta)$ as
\begin{equation}
Q(\boldsymbol\vartheta|\boldsymbol\vartheta^{(old)})=Q(\boldsymbol\eta|\boldsymbol\eta^{(old)})+Q(\boldsymbol\theta|\boldsymbol\theta^{(old)}),
\end{equation}
where
\begin{align}
Q(\boldsymbol\eta|\boldsymbol\eta^{(old)})&=\sum_{i=1}^n\sum_{k=1}^K\gamma^{(new)}_{ik}\log\{\pi_k(\mathbf{t}_i|\boldsymbol\eta)\},\nonumber\\
Q(\boldsymbol\theta|\boldsymbol\theta^{(old)})&=\sum_{i=1}^n\sum_{k=1}^K\gamma^{(new)}_{ik}\left[
\begin{aligned}
&-\frac{1}{2}\log(2\pi\sigma_k)-\frac{1}{2}v^{(new)}_{ik,1}
-\frac{v^{(new)}_{ik,2}\{y_i-\mu_k(\mathbf{x}_i;\boldsymbol\beta_k)\}^{2}}{2\sigma_k}\\
&+\frac{\alpha_k\{y_i-\mu_k(\mathbf{x}_i;\boldsymbol\beta_k)\}}{\sigma_k}
-\frac{\alpha^2_k}{2\sigma_k}v^{(new)}_{ik,3}-v^{(new)}_{ik,3}
\end{aligned}
\right].
\end{align}
where
\begin{eqnarray}
\gamma^{(new)}_{ik}&=&\mathbb{E}_{\boldsymbol\vartheta^{(old)}}\left\{z_{ik}|\mathbf{x}_i,\mathbf{t}_i,y_i\right\}\label{resp1}\\
v^{(new)}_{ik,1}&=&\mathbb{E}_{\boldsymbol\vartheta^{(old)}}\left\{\log v_i|\mathbf{x}_i,\mathbf{t}_i,y_i,z_{ik}=1\right\}\label{resp2}\\
v^{(new)}_{ik,2}&=&\mathbb{E}_{\boldsymbol\vartheta^{(old)}}\left\{1/v_i|\mathbf{x}_i,\mathbf{t}_i,y_i,z_{ik}=1\right\}\label{resp3}\\
v^{(new)}_{ik,3}&=&\mathbb{E}_{\boldsymbol\vartheta^{(old)}}\left\{v_i|\mathbf{x}_i,\mathbf{t}_i,y_i,z_{ik}=1\right\}\label{resp4}
\end{eqnarray}
The conditional expectation \eqref{resp1} is the well-known responsibility that the $k^{th}$ expert component takes for explaining the $i^{th}$ data point. This is given by the posterior probability
\begin{eqnarray}
\gamma^{(new)}_{ik}=\frac{\pi_k(\mathbf{t}_i|\boldsymbol\eta^{(old)})g(y_i|\alpha^{(old)}_k,\sigma^{(old)}_k,\mu_{k}(\mathbf{x}_i;\boldsymbol{\beta}^{(old)}_k))}{f(y_i|\mathbf{x}_i,\mathbf{t}_i;\boldsymbol\vartheta^{(old)})}
\end{eqnarray}
Notice that the conditional expectation $v^{(new)}_{ik,1}$ \eqref{resp2} is a constant in the maximisation of $Q(\boldsymbol\vartheta\mid\boldsymbol\vartheta^{(old)})$ with respect to $\boldsymbol\vartheta$. Therefore, we do not have to calculate it since it is not relevant for the purpose of obtaining the estimate of $\boldsymbol\vartheta$. However, it is important to note that it has a closed-form expression. This was also noted in \cite{franczak2014}. As in \cite{franczak2014}, the conditional expectations $v^{(new)}_{ik,2}$ and $v^{(new)}_{ik,3}$ are defined as follows
\begin{eqnarray}
v^{(new)}_{ik,2}&=&\sqrt{a^{(old)}_k/b^{(old)}_{ik}}R_\nu\left(\sqrt{a^{(old)}_kb^{(old)}_{ik}}\right)-\frac{2\nu}{b^{(old)}_{ik}}\\
v^{(new)}_{ik,3}&=&\sqrt{b^{(old)}_{ik}/a^{(old)}_{k}}R_\nu\left(\sqrt{a^{(old)}_kb^{(old)}_{ik}}\right)
\end{eqnarray}
where $\nu=1/2$, $a^{(old)}_k=2+(\alpha^{(old)}_k)^2/\sigma^{(old)}_k$, $b^{(old)}_{ik}=\{y_i-\mu_k(\mathbf{x}_i;\boldsymbol\beta^{(old)}_k)\}^2/\sigma^{(old)}_k$, for $i=1,2,\dots,n$ and $k=1,2,\dots,K$, and $R_\nu(z)=K_{\nu+1}(z)/K_{\nu}(z)$. These expressions are used when $b^{(old)}_{ik}>0$; in numerical implementations, exact zero residuals are replaced by a small positive tolerance. \\
\subsubsection*{M-step}
In the M-step of the EM-MM algorithm, we update the model parameters $\boldsymbol\theta^{(old)}$ by maximizing $Q(\boldsymbol\theta|\boldsymbol\theta^{(old)})$ with respect to $\boldsymbol\beta_k$, $\alpha_k$ and $\sigma_k$, for $k=1,2,\dots,K$, to obtain the respective updating equations as 
\begin{align}
\boldsymbol\beta^{(new)}_k
&=\left[\sum_{i=1}^n\gamma^{(new)}_{ik}v^{(new)}_{ik,2}\mathbf{x}_i\mathbf{x}^\top_i
-\frac{\sum_{i=1}^n\gamma^{(new)}_{ik}\mathbf{x}_i\left(\sum^{n}_{j=1}\gamma_{jk}^{(new)}\mathbf{x}^\top_j\right)}{\sum_{j=1}^n\gamma^{(new)}_{jk}v^{(new)}_{jk,3}}\right]^{-1}\nonumber\\
&\quad\times\left[\sum_{i=1}^n\gamma^{(new)}_{ik}v^{(new)}_{ik,2}\mathbf{x}_iy_i
-\frac{\sum_{i=1}^n\gamma^{(new)}_{ik}\mathbf{x}_i\left(\sum^{n}_{j=1}\gamma_{jk}^{(new)}y_j\right)}{\sum_{j=1}^n\gamma^{(new)}_{jk}v^{(new)}_{jk,3}}\right],\\
\alpha^{(new)}_k
&=\frac{\sum_{i=1}^n\gamma^{(new)}_{ik}\left(y_i-\mathbf{x}^\top_i\boldsymbol\beta^{(new)}_k\right)}{\sum_{i=1}^n\gamma^{(new)}_{ik}v^{(new)}_{ik,3}},\\
\sigma^{(new)}_k
&=\frac{\sum^{n}_{i=1}\gamma^{(new)}_{ik}v^{(new)}_{ik,2}\left(y_i-\mathbf{x}^\top_i\boldsymbol\beta^{(new)}_k\right)^2}{\sum_{i=1}^n\gamma^{(new)}_{ik}}\nonumber\\
&\quad-\frac{2\alpha^{(new)}_k\sum^{n}_{i=1}\gamma^{(new)}_{ik}\left(y_i-\mathbf{x}^\top_i\boldsymbol{\beta}^{(new)}_k\right)}{\sum_{i=1}^n\gamma^{(new)}_{ik}}\nonumber\\
&\quad+\frac{\alpha^{2(new)}_k\sum^{n}_{i=1}\gamma^{(new)}_{ik}v^{(new)}_{ik,3}}{\sum_{i=1}^n\gamma^{(new)}_{ik}}.
\end{align}

\subsubsection*{MM-step}
In the MM-step, we maximise $Q(\boldsymbol{\eta}\mid\boldsymbol{\eta}^{(old)})$ with
respect to the gating parameters $\boldsymbol{\eta}=(\boldsymbol{\eta}_{1},\dots,\boldsymbol{\eta}_{K})$
under the identifiability constraint $\boldsymbol{\eta}_{K}=\mathbf{0}$.
Writing $\tilde{\boldsymbol{\eta}}=(\boldsymbol{\eta}_{1}^{\top},\dots,\boldsymbol{\eta}_{K-1}^{\top})^{\top}\in\mathbb{R}^{(K-1)(q+1)}$
for the free parameters, we consider the concave objective function
\begin{equation}
Q(\tilde{\boldsymbol{\eta}}\mid\boldsymbol{\eta}^{(old)})=\sum_{i=1}^{n}\sum_{k=1}^{K}\gamma_{ik}^{(new)}\log\{\pi_{k}(\mathbf{t}_{i}\mid\boldsymbol{\eta})\},\label{Qeta_def}
\end{equation}
where $\pi_{k}(\mathbf{t}_{i}\mid\boldsymbol{\eta})$ is defined in
(\ref{logistic}) with $\boldsymbol{\eta}_{K}=\mathbf{0}$. Note that,
in the present M-step, the weights $\gamma_{ik}^{(new)}$ are fixed,
coming from the E-step.\\
Let $\mathbf{T}\in\mathbb{R}^{n\times(q+1)}$ be the design matrix
with $i$th row $\mathbf{t}_{i}^{\top}$, $\mathbf{E}^{(old)}=[\boldsymbol{\eta}_{1}^{(old)},\dots,\boldsymbol{\eta}_{K-1}^{(old)}]\in\mathbb{R}^{(q+1)\times(K-1)}$
and $\mathbf{G}^{(old)}=\mathbf{T}^{\top}(\boldsymbol{\Gamma}-\boldsymbol{\Pi}^{(old)})\in\mathbb{R}^{(q+1)\times(K-1)}$,
where $\boldsymbol{\Pi}^{(old)}=[\mathbf{p}_{1}(\tilde{\boldsymbol{\eta}}^{(old)})^{\top};\dots;\mathbf{p}_{n}(\tilde{\boldsymbol{\eta}}^{(old)})^{\top}]$, whose $i^{th}$ row is $\mathbf{p}_{i}(\tilde{\boldsymbol{\eta}}^{(old)})^{\top}$, with $\mathbf{p}_{i}(\tilde{\boldsymbol{\eta}})=(\pi_{1}(\mathbf{t}_{i}\mid\boldsymbol{\eta}),\dots,\pi_{K-1}(\mathbf{t}_{i}\mid\boldsymbol{\eta}))^{\top}$, and $\boldsymbol{\Gamma}=[\boldsymbol{\gamma}_{1}^{\top};\dots;\boldsymbol{\gamma}_{n}^{\top}]$ with $\boldsymbol{\gamma}_{i}=(\gamma_{i1}^{(new)},\dots,\gamma_{i,K-1}^{(new)})^{\top}$ being the corresponding responsibilities.\\
The MM update equation is 
\begin{equation}
\mathbf{E}^{(new)}=\mathbf{E}^{(old)}+2(\mathbf{T}^{\top}\mathbf{T})^{-1}\mathbf{G}^{(old)}(\mathbf{I}_{K-1}+\mathbf{1}\mathbf{1}^{\top}).\label{LB_update_matrix_one}
\end{equation}
The derivation of the MM update equation \eqref{LB_update_matrix_one} can be found in Appendix \ref{MMstep}.

\noindent Let $\mathcal{M}(\tilde{\boldsymbol{\eta}}\mid\tilde{\boldsymbol{\eta}}^{(old)})$ be the minoriser of $Q(\tilde{\boldsymbol{\eta}}\mid\boldsymbol{\eta}^{(old)})$, then by construction, 
\[
Q(\tilde{\boldsymbol{\eta}}^{(new)}\mid\boldsymbol{\eta}^{(old)})\ge\mathcal{M}(\tilde{\boldsymbol{\eta}}^{(new)}\mid\tilde{\boldsymbol{\eta}}^{(old)})\ge\mathcal{M}(\tilde{\boldsymbol{\eta}}^{(old)}\mid\tilde{\boldsymbol{\eta}}^{(old)})=Q(\tilde{\boldsymbol{\eta}}^{(old)}\mid\boldsymbol{\eta}^{(old)}),
\]
so a single update (\ref{LB_update_matrix_one}) increases (or leaves unchanged)
the gating contribution to the $Q$-function. In the overall EM--MM
hybrid algorithm, we perform one update for M-step~1 (the
expert parameters), and then one update for the MM-step given
above (the gating parameters). Since each of these two updates is
chosen not to decrease the relevant part of $Q(\cdot\mid\cdot)$ with
the E-step weights fixed, the combined outer iteration is a valid
generalised EM step, and hence yields a nondecreasing sequence of
observed-data log-likelihood values.\\
The above EM-MM steps are repeated until convergence. 
\subsection{Model Selection}\label{subsec4}
In practice, prior to estimating the SALMoE model, we need to choose $K$, the number of expert components. Data-driven procedures are usually employed for the purpose of choosing the optimal value of $K$. A popular data-driven procedure for choosing $K$ is by using the likelihood-based information criteria (IC) (see Chapter 7 of \cite{fruhwirth2019}). In this paper, we propose to use the Bayesian information criterion (BIC) \cite{schwarz1978},
\begin{equation}
    \text{BIC}(K)=-2\times\ell(\hat{\boldsymbol\vartheta})+\log(n)\times \text{df}(K)
\end{equation}
and the integrated classification likelihood (ICL) \cite{biernacki2000},
\begin{equation}
    \text{ICL}(K)=-2\times\ell_{\mathrm{cl}}(\hat{\boldsymbol\vartheta})+\log(n)\times \text{df}(K),
\end{equation}
where $\ell(\hat{\boldsymbol\vartheta})$ is the value of the observed-data log-likelihood at convergence of the EM-MM algorithm and $\ell_{\mathrm{cl}}(\hat{\boldsymbol\vartheta})$ denotes the classification log-likelihood evaluated at the MAP allocations. Specifically, let $\hat\gamma_{ik}$ be the fitted posterior responsibility in \eqref{resp1} evaluated at $\hat{\boldsymbol\vartheta}$. Then
\begin{equation}\label{MAP}
\hat{z}_{ik}^{\mathrm{MAP}}=\begin{cases}
1, & \text{if } k=\min\operatorname*{arg\,max}_{1\le j\le K}\hat\gamma_{ij},\\
0, & \text{otherwise},
\end{cases}
\end{equation}
and
\[
\ell_{\mathrm{cl}}(\hat{\boldsymbol\vartheta})=\sum_{i=1}^{n}\sum_{k=1}^{K}\hat{z}_{ik}^{\mathrm{MAP}}\log\left\{\pi_k(\mathbf{t}_i\mid\hat{\boldsymbol\eta})g(y_i\mid\hat\alpha_k,\hat\sigma_k,\mu_k(\mathbf{x}_i;\hat{\boldsymbol\beta}_k))\right\}.
\]
The use of $\min$ in \eqref{MAP} gives a deterministic rule for resolving ties. Also, $\text{df}(K)$ is the number of estimated parameters in the fitted model. For the SALMoE, $\text{df}(K)=K(p+q+4)-q-1$.\\
In addition to the BIC and ICL, we use the PanIC framework of \citet{Nguyen2024PanIC}, which yields consistent order selection under easily verifiable regularity conditions. In the likelihood setting, the PanIC is defined as
\begin{equation}\label{panIC}
\text{PanIC}(K)=-2\ell(\hat{\boldsymbol{\vartheta}}_{K})+2nP_{K,n},
\end{equation}
where 
\[
P_{K,n}=\alpha\text{df}(K)\frac{\log_{+}^{(\beta)}(n)}{\sqrt{n}},\qquad\alpha>0,
\]
is a penalty with
\[
\log_{+}^{(\beta)}(n)=\underbrace{\log_{+}\circ\log_{+}\circ\cdots\circ\log_{+}}_{\beta\ \text{times}}(n),\qquad\beta\in\mathbb{N},
\]
and $\log_{+}(x)=\max\{1,\log(x)\}$.\\
Note that the PanIC requires a choice of the values of the constants $\alpha$ and $\beta$. As in \cite{Nguyen2024PanIC}, our choice of $\alpha$ is coupled with the BIC. Thus, we choose $\alpha=\alpha(\beta,\nu)>0$ which is a value of $\alpha$ that makes the PanIC with order $\beta$ to be equal to the BIC when the sample size is $\nu$, defined as:
\begin{equation}\label{pan_calibrate}
    \alpha(\beta,\nu)=\frac{\log(\nu)}{2\sqrt{\nu}\log_{+}^{(\beta)}(\nu)}
\end{equation}
For more details about the implementation of the PanIC, see our numerical experiments in Section \ref{choose_k}.\\
For more details about the PanIC for the proposed SALMoE model, see Appendix \ref{PanIC}.\\
The optimal value of $K$ is the one that corresponds to the minimum values of the $\text{BIC}(K)$, $\text{ICL}(K)$ or $\text{PanIC}(K)$ over a range of values of $K\in\{1,2,\dots,K_{max}\}$, where $K_{max}$ is the largest value of $K$ considered.
\subsection{Prediction using the SALMoE model}\label{subsec5}
In this section, we obtain the predictive distribution of the SALMoE model which can be used to make predictions about the response variable $y$ given new values of the covariates $(\mathbf{x},\mathbf{t})$. We also present some useful statistical measures to assess the predictive performance of the SALMoE. 
\subsubsection{Predictive distribution of the SALMoE}
The predictive distribution is a probability distribution of the response variable $y$, which is obtained by substituting the MLE $\hat{\boldsymbol\vartheta}$ into \eqref{model2} to give
\begin{eqnarray}\label{pred_dist}
f(y|\mathbf{x},\mathbf{t};\hat{\boldsymbol\vartheta})=\sum_{k=1}^K\pi_k(\mathbf{t}|\hat{\boldsymbol\eta})f_k(y|\mathbf{x};\hat{\boldsymbol\theta}_k)
\end{eqnarray}
From \citet{chamroukhi2016,chamroukhi2017}, the mean of the predictive distribution \eqref{pred_dist} is given by
\begin{equation}
\mu_{\hat{\boldsymbol\vartheta}}
=\mathbb{E}_{\hat{\boldsymbol\vartheta}}[Y|\mathbf{x},\mathbf{t}]
=\sum_{k=1}^K\pi_k(\mathbf{t}|\hat{\boldsymbol\eta})\mathbb{E}_{\hat{\boldsymbol\vartheta}}[Y|Z=k,\mathbf{x}].\label{exp_MoE}
\end{equation}
The corresponding variance is
\begin{align}
\sigma^2_{\hat{\boldsymbol\vartheta}}
&=\mathbb{V}_{\hat{\boldsymbol\vartheta}}[Y|\mathbf{x},\mathbf{t}]\nonumber\\
&=\sum_{k=1}^K\pi_k(\mathbf{t}|\hat{\boldsymbol\eta})\left[(\mathbb{E}_{\hat{\boldsymbol\vartheta}}[Y|Z=k,\mathbf{x}])^2+\mathbb{V}_{\hat{\boldsymbol\vartheta}}(Y|Z=k,\mathbf{x})\right]
-\mu^2_{\hat{\boldsymbol\vartheta}}.\label{var_MoE}
\end{align}
where $\mathbb{E}_{\hat{\boldsymbol\vartheta}}[Y|Z=k,\mathbf{x}]$ and $\mathbb{V}_{\hat{\boldsymbol\vartheta}}[Y|Z=k,\mathbf{x}]$ are the mean and variance, respectively, of the $k^{th}$ expert component defined as follows for the SALMoE model
\begin{eqnarray}
\mathbb{E}_{\hat{\boldsymbol\vartheta}}[Y|Z=k,\mathbf{x}]&=&\mathbf{x}^\top\hat{\boldsymbol\beta}_k+\hat{\alpha}_k\label{exp_SALMoE}\\
\mathbb{V}_{\hat{\boldsymbol\vartheta}}[Y|Z=k,\mathbf{x}]&=&\hat{\alpha}^2_k+\hat{\sigma}_k\label{var_SALMoE}
\end{eqnarray}
Substituting \eqref{exp_SALMoE} and \eqref{var_SALMoE} into \eqref{exp_MoE} and \eqref{var_MoE}, respectively, we obtain the mean and variance of the SALMoE model.\\
For given values of the covariates $\mathbf{x}$ and $\mathbf{t}$, the predicted value of the response variable $y$, denoted $\hat{y}$, is given by the mean of the SALMoE model \cite{chamroukhi2016}, that is $\hat{y}=\mathbb{E}_{\hat{\boldsymbol\vartheta}}[Y|\mathbf{x},\mathbf{t}]$.\\
However, from a probabilistic perspective, a predictive distribution is typically desirable because it expresses our uncertainty about the predicted value of the response variable $y$ for each of the values of the covariates $\mathbf{x}$ and $\mathbf{t}$ \cite{bishop2006}.\\
Consider a random sample $\{(y_i,\mathbf{x}_i,\mathbf{t}_i):i=1,2,\dots,n\}$. To quantify the uncertainty of our predictions, we will use the following approximate $95\%$ pointwise predictive interval given by
\begin{equation}
    \hat{y}_i\pm 2\times\text{sd}_{\hat{\boldsymbol\vartheta}}(y_i|\mathbf{x}_i,\mathbf{t}_i)\quad \text{for }i=1,2,\dots,n
\end{equation}
where $\text{sd}_{\hat{\boldsymbol\vartheta}}(y_i|\mathbf{x}_i,\mathbf{t}_i)=\sqrt{\mathbb{V}_{\hat{\boldsymbol\vartheta}}[y_i|\mathbf{x}_i,\mathbf{t}_i]}$, for $i=1,2,\dots,n$, is the pointwise standard deviation of the predictive distribution of the SALMoE model.
\subsection{Model-based clustering using the SALMoE model}\label{subsec6}
The SALMoE model can be used to obtain a partition of an observed dataset $\{(\mathbf{x}_i,\mathbf{t}_i,y_i):i=1,2,\dots,n\}$ into $K$ clusters. Each cluster is associated with a given expert component. Given $\hat{\boldsymbol{\vartheta}}$ obtained at convergence of the EM-MM algorithm, the posterior probability
\begin{eqnarray}
\hat\gamma_{ik}=\frac{\hat\pi_k(\mathbf{t}_i|\hat{\boldsymbol\eta})g(y_i|\hat\alpha_k,\hat\sigma_k,\mu_k(\mathbf{x}_i;\hat{\boldsymbol{\beta}}_k))}{f(y_i|\mathbf{x}_i,\mathbf{t}_i;\hat{\boldsymbol\vartheta})}
\end{eqnarray}
provides a soft allocation of the data point $(\mathbf{x}_i,\mathbf{t}_i,y_i)$ into the $k^{th}$ cluster. The corresponding hard allocation is the MAP allocation $\hat{z}_{ik}^{\mathrm{MAP}}$ defined in \eqref{MAP}, which estimates the latent component indicator variable $z_{ik}$.\\
\section{Simulation study}\label{sec3}
In this section, we perform an extensive numerical study:
\begin{enumerate}
    \item to evaluate the performance of the proposed EM-MM algorithm for fitting the proposed SALMoE model;
    \item to demonstrate the robustness of the SALMoE model compared to the GMoE model.% and STMoE model \cite{chamroukhi2017};
    \item to identify the most suitable information criteria for choosing the number of expert components for the SALMoE model; and
    \item to evaluate the clustering capability of the SALMoE using the approach in Section \ref{subsec6}.
\end{enumerate}
\paragraph{Initialization}
The parameters of the SALMoE model, $\boldsymbol{\vartheta}$, are initialized using a similar approach outlined in \cite{chamroukhi2016,chamroukhi2017}. The data are randomly partitioned into $K$ groups and a GMoE model with $K$ components is fitted to the data. We use the estimates of the component standard deviations to initialize the scale parameters $\sigma_k$, for $k=1,2,\dots,K$. The skewness parameters $\alpha_k$, for $k=1,2,\dots,K$, are initialized similarly to the skewness parameters of the STMoE model of \cite{chamroukhi2017}. This process is repeated for 30 times and the parameter estimates that produce the largest likelihood value are used to initialize the algorithm.
\paragraph{Stopping rule}
To assess convergence, we use the relative log-likelihood increment
\begin{equation}\label{converge_crit}
\frac{\ell(\boldsymbol\vartheta^{(new)})-\ell(\boldsymbol\vartheta^{(old)})}{|\ell(\boldsymbol\vartheta^{(old)})|}.
\end{equation}
Convergence is declared whenever \eqref{converge_crit} is less than a prespecified threshold $\epsilon$. In our simulations and real data analysis, we set $\epsilon=10^{-5}$.
\subsection{Estimation performance}
We consider two scenarios to demonstrate the performance of the EM-MM algorithm. In each scenario, we generate $100$ samples of sizes $n=100, 500, 1000$ and $2000$ from a $K=2$ component SALMoE model. To measure the accuracy of the estimated parameters, we make use of the mean squared error (MSE) and BIAS, respectively,
\begin{eqnarray}
    \text{MSE}(\hat\theta_k)=(\hat\theta_k-\theta_k)^2~~~~\text{and}~~~~~\text{BIAS}(\hat\theta_k)=(\hat\theta_k-\theta_k)
\end{eqnarray}
where $\hat\theta_k$ and $\theta_k$ are the estimated and true parameter values, respectively.\\
These measures are averaged over the $100$ replications.
\subsubsection{Scenario 1}
In this scenario, the response values $y_i$, for $i=1,2,\dots,n$, are generated from a $K=2$ component SALMoE model \eqref{SALMoE}. The covariates $(\mathbf{x}_i,\mathbf{t}_i)$ are generated such that $\mathbf{x}_i=\mathbf{t}_i=(1,x_i)^\top$, where $x_i$ is generated from the uniform distribution on the interval $(-1,1)$. The parameter values used for this scenario are given in Table \ref{tab:sim_params1}. 
\begin{table}[]
    \centering
    \caption{Parameter values used in the simulation}
    \label{tab:sim_params1}
    \begin{tabular}{|c|c|c|c|c|c|}
    \hline
         &\multicolumn{4}{|c|}{Parameters}  \\
    \hline
        Component 1&$\boldsymbol{\eta}_1=(0,10)$&$\boldsymbol{\beta}_1=(0,1)$&$\alpha_1=1$&$\sigma_1=0.1$\\
        Component 2&$\boldsymbol{\eta}_2=(0,0)$&$\boldsymbol{\beta}_2=(0,-1)$&$\alpha_2=0.8$&$\sigma_2=0.1$\\
    \hline
    \end{tabular}
\end{table}
The results of the simulation are given in Table \ref{sim_res1}. We can see that the MSE generally decreases as the sample size increases, which is evidence in support of the effectiveness of the proposed EM-MM algorithm. 
\begin{table}[htbp]
\centering
\caption{Average Bias and MSE for fitted parameters of the SALMoE over the 100 samples in Scenario 1.}
\renewcommand{\arraystretch}{1.2}
\setlength{\tabcolsep}{2.5pt}

\begin{tabular}{c c cc cc cc cc cc}
\toprule
\multirow{2}{*}{Sample Size} & \multirow{2}{*}{Measure} 
& \multicolumn{10}{c}{Parameters} \\
\cmidrule(lr){3-12}
&&$\eta_{10}$&$\eta_{11}$&$\beta_{10}$&$\beta_{11}$&$\beta_{20}$&$\beta_{21}$&$\sigma_1$&$\sigma_2$&$\alpha_1$&$\alpha_2$\\
\midrule

\multirow{2}{*}{$n=100$}
& MSE&2.1218&41.1316 &0.6960&0.5399&0.8455&0.5923&0.0897&0.1684&0.0726&0.4333\\
& BIAS&-0.0975&4.7822&-0.6336&0.5178&-0.7097&-0.5594&-0.1714&-0.2496& 0.1724&0.5257\\
\multirow{2}{*}{$n=500$}
& MSE &0.1769 &8.2204& 0.1152&0.0837&0.0661&0.0712& 0.0197&0.0061&0.0135&0.0736\\
& BIAS  &-0.0185 &1.0736&-0.1086&0.0853&-0.0791&-0.0699&-0.0348&-0.0158&0.0202 &0.0734 \\
\multirow{2}{*}{$n=1000$}
& MSE &0.1100&2.9331&0.0157&0.0212&0.0344&0.0223&0.0009&0.0091&0.0044&0.0208\\
& BIAS  &0.0100 &0.3331&-0.0258&0.0234&-0.0288&-0.0224&-0.0013&-0.0109& 0.0097 &0.0145\\
\multirow{2}{*}{$n=2000$}
& MSE &0.0499 &1.7055&0.0069&0.0101&0.0164&0.0119&0.0004&0.0040&0.0026&0.0096   \\
& BIAS  &0.0067 &0.3974&-0.0123&0.0151&-0.0159&-0.0146&-0.0014&-0.0036&0.0024 &0.0108\\
\bottomrule
\end{tabular}
\label{sim_res1}
\end{table}
In addition to the results in Table \ref{sim_res1}, Figure \ref{fig:sim1} gives a plot of the estimated mean function, variance function, gating functions and component mean functions along with their true counterparts for sample size $n=500$. The estimated functions are obtained by averaging over the 100 fitted functions. As can be clearly seen from the figure, the estimated functions are in line with the true functions which provides further evidence of the effectiveness of the proposed modelling framework. Lastly, Figure \ref{fig:loglik} plots the observed log likelihood values at each iteration of the EM-MM algorithm for a single instance with $n=500$. It can be seen from the figure that the likelihood values are increasing at each iteration.
\begin{figure}[htbp]
    \centering
    \begin{subfigure}[b]{0.45\textwidth}
        \includegraphics[width=\textwidth]{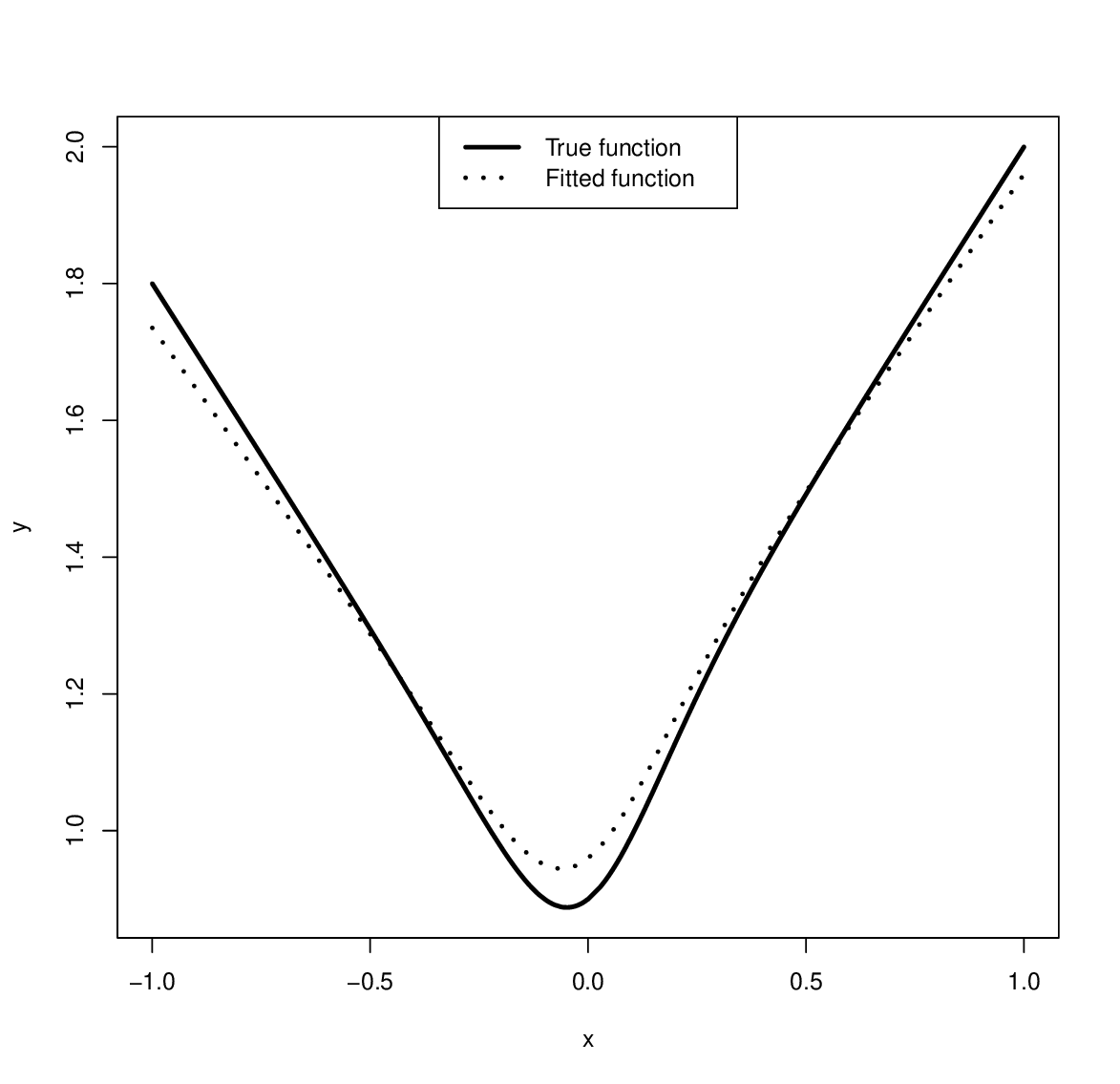}
        \caption{Mean function}
    \end{subfigure}
    \hfill
    \begin{subfigure}[b]{0.45\textwidth}
        \includegraphics[width=\textwidth]{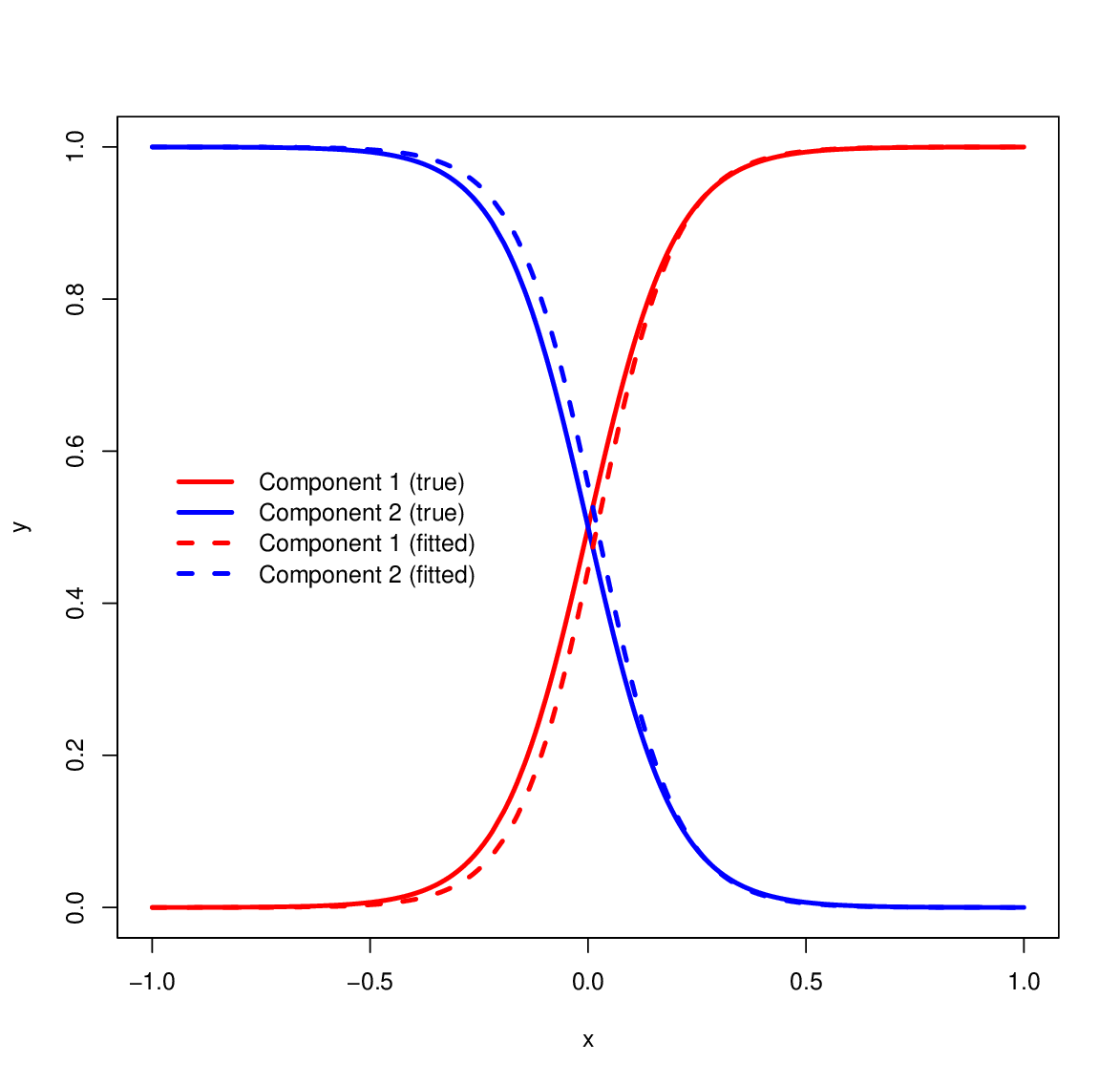}
        \caption{Mixing proportion functions}
    \end{subfigure}
    \vspace{1cm}
    \begin{subfigure}[b]{0.45\textwidth}
        \includegraphics[width=\textwidth]{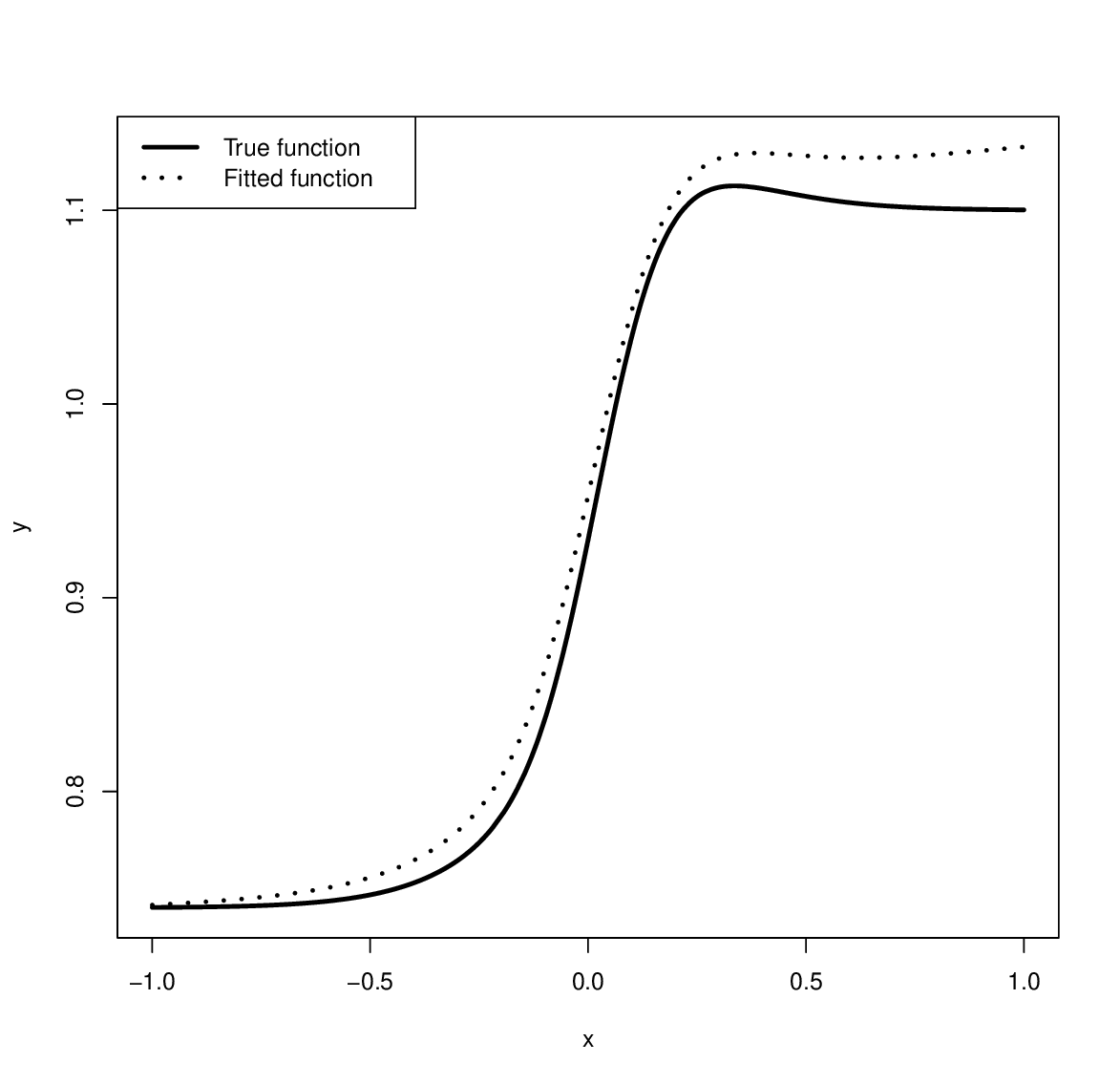}
        \caption{Variance function}
    \end{subfigure}
    \hfill
    \begin{subfigure}[b]{0.45\textwidth}
        \includegraphics[width=\textwidth]{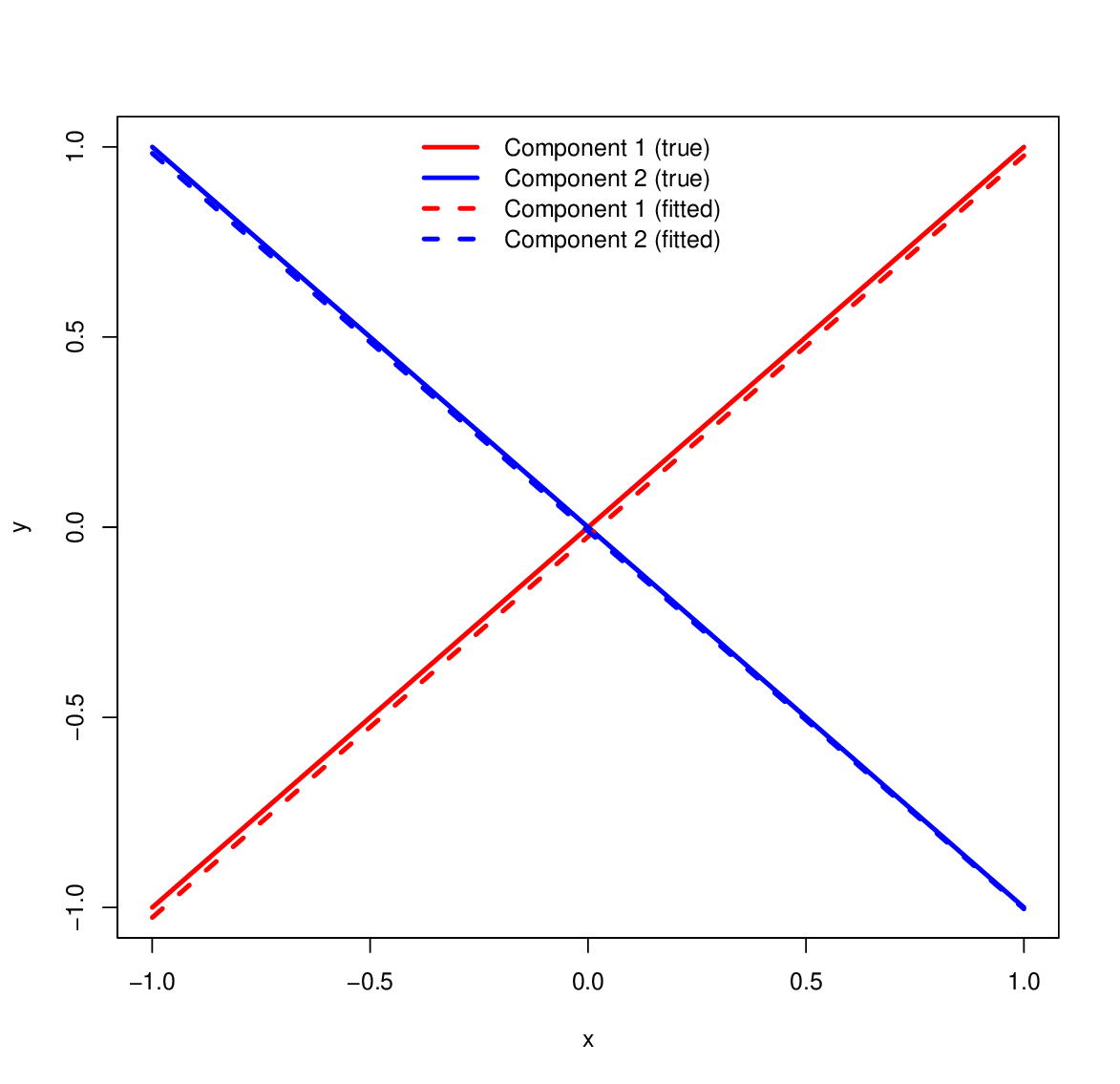}
        \caption{Component mean functions}
    \end{subfigure}
    \caption{Plot of the estimated functions obtained from averaging over the 100 fitted SALMoE models. For each estimated function, we also plot its corresponding true function.}
    \label{fig:sim1}
\end{figure}

\begin{figure}
    \centering
    \includegraphics[width=0.9\textwidth]{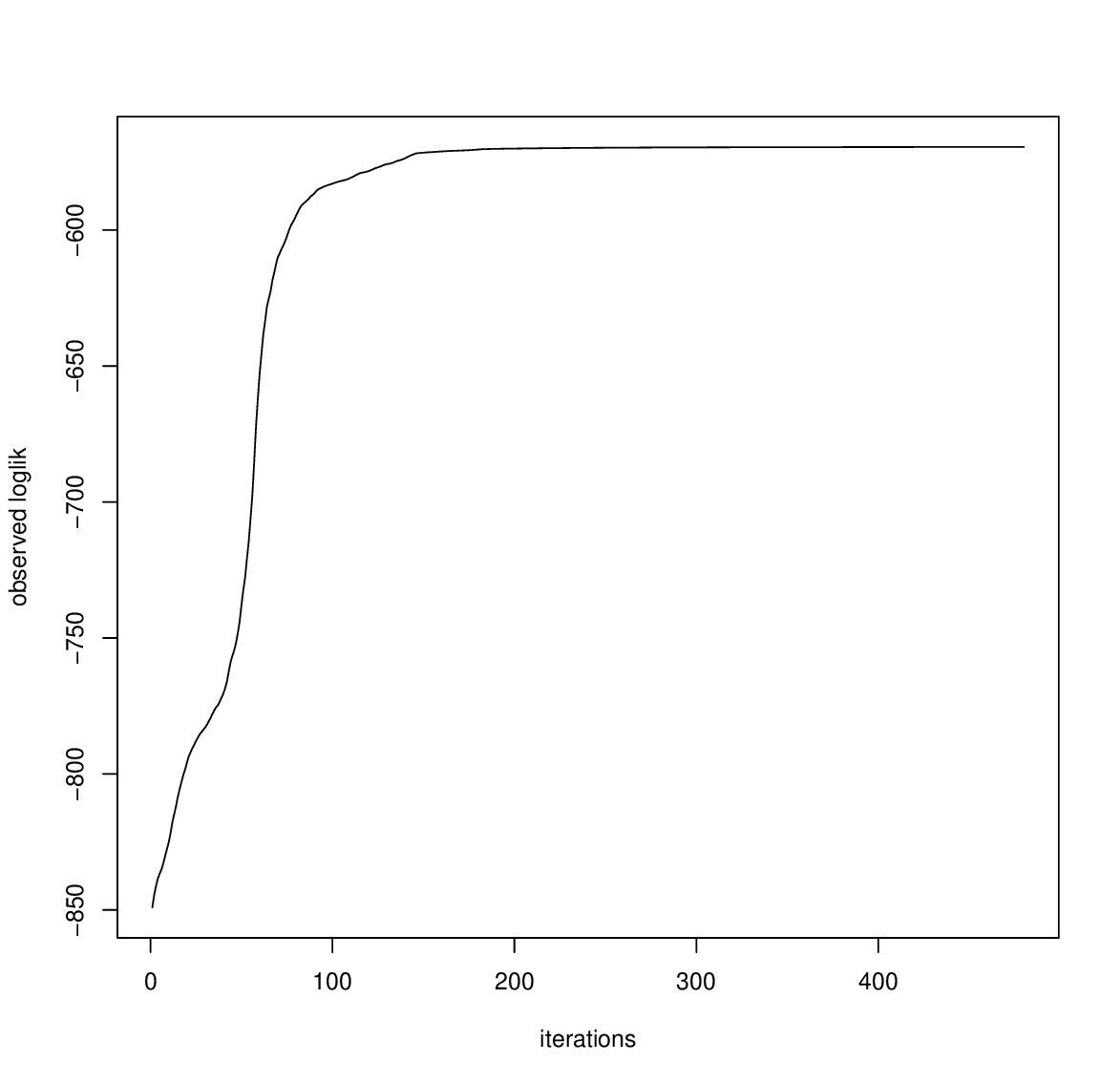}
    \caption{Log likelihood values at each iteration of the EM-MM algorithm for a single instance with $n=500$.}
    \label{fig:loglik}
\end{figure}
\subsubsection{Scenario 2}
In this scenario, the response values $y_i$, for $i=1,2,\dots,n$, are generated from the $K=2$ component SALMoE model \eqref{SALMoE}. The covariates $(\mathbf{x}_i,\mathbf{t}_i)$ are generated such that $\mathbf{x}_i=\mathbf{t}_i=(1,x_{i1},x_{i2},x_{i3})^\top$, where $x_{i1}$, $x_{i2}$ and $x_{i3}$ are generated independently from the uniform distribution on the interval $(-1,1)$. The parameter values used for this scenario are $\boldsymbol{\eta}_1=(0,5,-2,10)$, $\boldsymbol{\beta}_1=(0,-1,0.5,1)$ and $\boldsymbol{\beta}_2=(0,1,0.5,-1)$. The rest of the parameters are the same as in Scenario 1.
Table \ref{sim_res2} gives the results of the simulation. The results show that the proposed method maintains its good performance even when we have multiple covariates. This demonstrates the scalability of the proposed method.
\begin{sidewaystable}[htbp]
\centering
\caption{Average Bias and MSE for the fitted parameters of the SALMoE model over the 100 samples in Scenario 2.}
\renewcommand{\arraystretch}{1.2}
\setlength{\tabcolsep}{2.5pt}

\begin{tabular}{c cccc cccccc cc cc cc cc}
\toprule
\multirow{2}{*}{Sample Size} & \multirow{2}{*}{Measure} 
& \multicolumn{16}{c}{Parameters} \\
\cmidrule(lr){3-18}
&&$\eta_{10}$&$\eta_{11}$&$\eta_{12}$&$\eta_{13}$&$\beta_{10}$&$\beta_{11}$&$\beta_{12}$&$\beta_{13}$&$\beta_{20}$&$\beta_{21}$&$\beta_{22}$&$\beta_{23}$&$\sigma_1$&$\sigma_2$&$\alpha_1$&$\alpha_2$\\
\midrule

\multirow{2}{*}{$n=100$}
& MSE&1.2711 &14.2313  &3.0418 &16.8822 &0.3019 &0.1985 &0.0503&0.2258&0.2466  &0.1136 &0.0355 &0.2007&0.1259 &0.1815 &0.0899 &0.3383\\
& BIAS&-0.3768&-0.1032 &-0.0099  &0.1800 &-0.2014 &-0.2292 & 0.2115  &0.4231&  0.1087 & 1.3790  &0.0548  &1.0768 &-0.1087 &-1.3790& -0.0548 &-1.0768\\
\multirow{2}{*}{$n=500$}
& MSE &0.0862& 0.9980 &0.3578&2.3326 &0.0033& 0.0016& 0.0011 &0.0038& 0.0028 &0.0015& 0.0016 &0.0031 &0.0023 &0.0020& 0.0059 &0.0035\\
& BIAS  & -0.0154  &0.0071  &0.0087& -0.0017 &-0.0027& -0.0017&  0.0142 & 0.0119 & 0.0127&  0.3336 &-0.0556 & 0.3604 &-0.0127 &-0.3336  &0.0556&-0.3604\\
\multirow{2}{*}{$n=1000$}
& MSE &0.0493& 0.4550 &0.1406 &0.9844 &0.0006& 0.0006 &0.0006 &0.0012 &0.0008 &0.0007 &0.0006& 0.0014& 0.0006 &0.0006& 0.0031& 0.0021\\
& BIAS &-0.0048& -0.0025 &-0.0002& -0.0038 & 0.0036 & 0.0018  &0.0087 & 0.0010& -0.0104 &0.1693 &-0.0358 & 0.2277  &0.0104 &-0.1693 & 0.0358& -0.2277\\
\multirow{2}{*}{$n=2000$}
& MSE &0.0206& 0.2417 &0.0923 &0.5536 &0.0005 &0.0004 &0.0003 &0.0006 &0.0003& 0.0002 &0.0002 &0.0004 &0.0004 &0.0003 &0.0014 &0.0011  \\
& BIAS &-0.0038 &0.0024& -0.0034 & 0.0021 & 0.0029 &0.0009  &0.0027&  0.0010&  0.0054  &0.1623 &-0.0673 &0.2752&-0.0054&-0.1623 &0.0673&-0.2752\\
\bottomrule
\end{tabular}
\label{sim_res2}
\end{sidewaystable}
\subsection{Robustness performance}
We consider two scenarios to examine the robustness of the SALMoE model in the presence of outliers and skewness compared to the GMoE model. To measure the goodness of the fitted model, we calculate the root MSE (RMSE) of the fitted mean function
\begin{equation}\label{mse_mean}
    \left(\frac{1}{n}\sum_{i=1}^n\left(\mathbb{E}_{\boldsymbol\vartheta}[Y_i|\mathbf{x}_i,\mathbf{t}_i]-\mathbb{E}_{\hat{\boldsymbol\vartheta}}[Y_i|\mathbf{x}_i,\mathbf{t}_i]\right)^2\right)^{1/2}
\end{equation}
where $\mathbb{E}_{\boldsymbol\vartheta}[Y|\mathbf{x},\mathbf{t}]$ and $\mathbb{E}_{\hat{\boldsymbol\vartheta}}[Y|\mathbf{x},\mathbf{t}]$ are the true and fitted mean functions, respectively.
\subsubsection{Scenario 1}
In the first scenario, we generate $100$ samples of size $n=500$ from both the SALMoE model and the GMoE model using the parameter values in Table \ref{tab:sim_params1}. For the GMoE case, the same gating and regression parameters are used and the SAL skewness parameters are ignored. For each generated sample, we randomly select $c\%$ of the data points and substitute them with noisy data points obtained by generating $x$ from the uniform distribution on the interval $(-1,1)$ and setting $y$ to $-2$. We consider the following values for $c=1,2,3,4$ and $5$. The results of the experiment are given in Figure \ref{fig:robust1} for the data generated from the SALMoE model and Figure \ref{fig:robust2} for the data generated from the GMoE model. From the figures, we can see that the SALMoE model clearly outperforms the GMoE model in both cases. These results clearly highlight the robustness of the proposed SALMoE model.
\begin{figure}
    \centering
    \includegraphics[width=\textwidth]{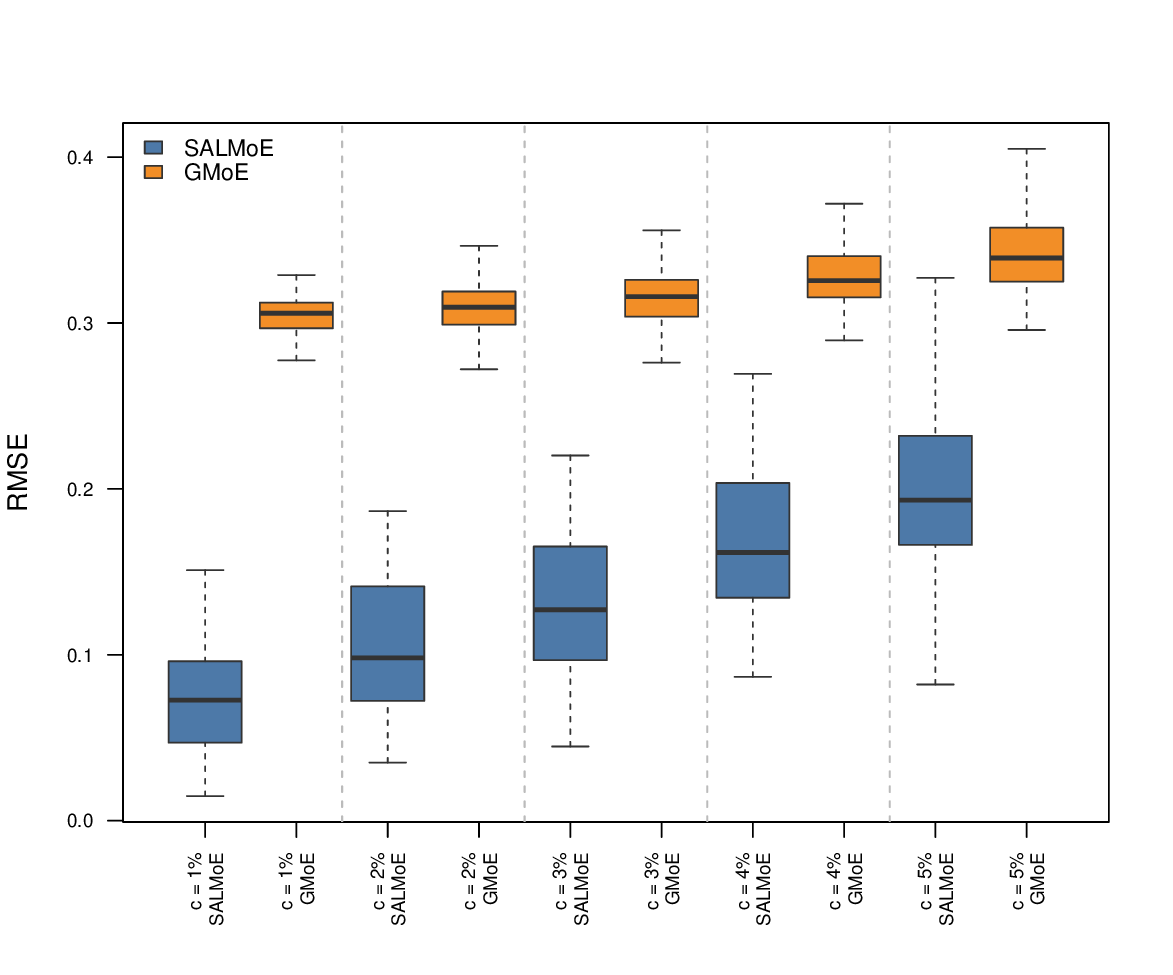}
    \caption{RMSE values for the SALMoE and GMoE models when the data are generated from the SALMoE model with $c\%$ of the data substituted with noise data points and $c=1,2,3,4$ and $5$.}
    \label{fig:robust1}
\end{figure}

\begin{figure}
    \centering
    \includegraphics[width=\textwidth]{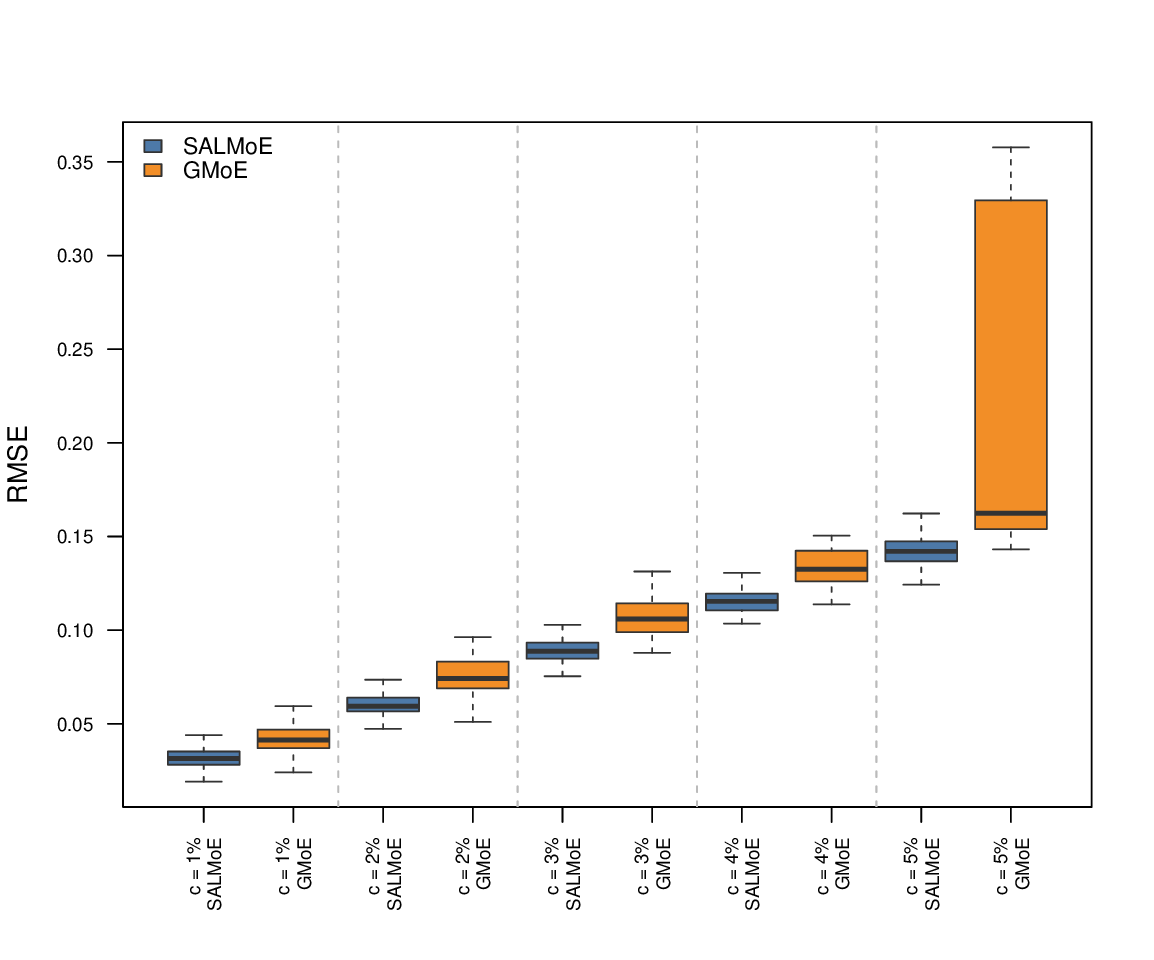}
    \caption{RMSE values for the SALMoE and GMoE models when the data are generated from the GMoE model with $c\%$ of the data substituted with noise data points and $c=1,2,3,4$ and $5$.}
    \label{fig:robust2}
\end{figure}
\subsubsection{Scenario 2}
In this scenario, we generate $100$ samples of size $n=500$ from a $K=2$ component MoE model in which each expert component follows a skew-normal (SN) distribution introduced by \citet{azzalini1985}. The density of the SN distribution is $\mathrm{SN}(y|\mu,\sigma,\lambda)$, where $\mu\in \mathbb{R}$, $\sigma\in\mathbb{R}^+$ and $\lambda\in\mathbb{R}$ are the location, scale and skewness parameters, respectively. The skewness parameter controls the shape of the distribution. The values $\lambda<0$, $\lambda=0$ and $\lambda>0$ correspond to the case of negative skewness, symmetric (normal) and positive skewness, respectively. We consider a variety of shapes specified by the skewness parameter values $\lambda=-20,-10,0,10$ and $20$. The remaining gating, regression and scale parameters are as given in Table \ref{tab:sim_params1}. The results of the experiment are given in Figure \ref{fig:skew}. The following three observations can be made from the figure: first, in general, the SALMoE model outperforms the GMoE model when the data exhibit asymmetric behaviour. Second, the outperformance of the SALMoE model is more pronounced when the data are positively skewed than when they are negatively skewed. Notably, the GMoE model is less stable when the data are positively skewed than when they are negatively skewed. Lastly, as expected, when $\lambda=0$, the GMoE model outperforms the SALMoE model. 

\begin{figure}
    \centering
    \includegraphics[width=\textwidth]{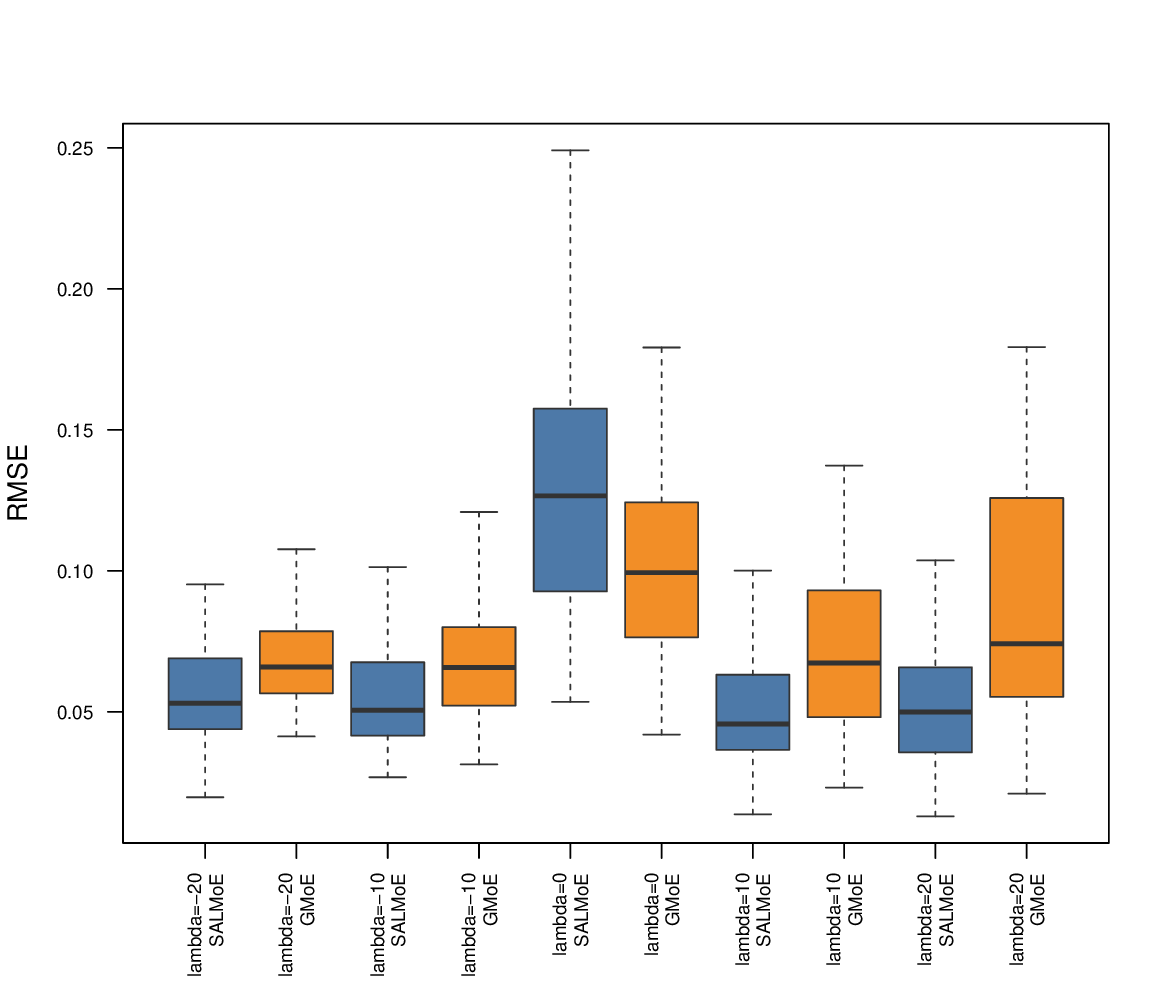}
    \caption{RMSE values for the SALMoE and GMoE models when the data are generated from the MoE model in which each expert component distribution follows a skew-normal distribution.}
    \label{fig:skew}
\end{figure}
\subsection{Model-based clustering performance}
In this experiment, we examine the clustering capability of the SALMoE model. We consider the $K=2$ component MoE model. We will compare the performance of the SALMoE model with the GMoE model under four scenarios: 
\begin{enumerate}
\item[(a)] the data are generated from a Gaussian distribution; 
\item[(b)] same as (a) but with $5\%$ of the data substituted with noisy data points; 
\item[(c)] the data are generated from a skew-normal (SN) distribution with skewness parameters $\lambda_1=\lambda_2=20$;
\item[(d)] same as (c) but with $5\%$ of the data substituted with noisy data points.
\end{enumerate}
Scenario (a) considers the case of symmetry with no outliers in the data. Scenario (b) considers the case of symmetric but heavy-tailed data. Scenario (c) considers the case of asymmetric data with lighter tails. Scenario (d) considers the case of asymmetric and heavy-tailed data.\\
In all these scenarios, the rest of the MoE parameter values are as given in Table \ref{tab:sim_params1}. To measure the performance of the models, we will make use of the misclassification error (ClassErr) and the adjusted Rand index (ARI) to assess the effectiveness of the clustering procedure. The results of this experiment are given in Figure \ref{fig:clus}.\\
From the figures, we can see that when the data are Gaussian, the GMoE model slightly outperforms the SALMoE model in terms of clustering performance. However, the clustering performance of the GMoE model deteriorates when we add noise and skewness to the data, second and third rows, respectively. The deterioration in the clustering performance of the GMoE model worsens in the case when we have both skewed and heavy-tailed data with outliers (last row). In contrast, in all four scenarios, the robustness of the SALMoE model can be clearly seen. This is especially true in Scenario (d).
\begin{figure}[p]
    \centering
    \begin{subfigure}[b]{0.35\textwidth}
        \centering
        \includegraphics[width=\textwidth]{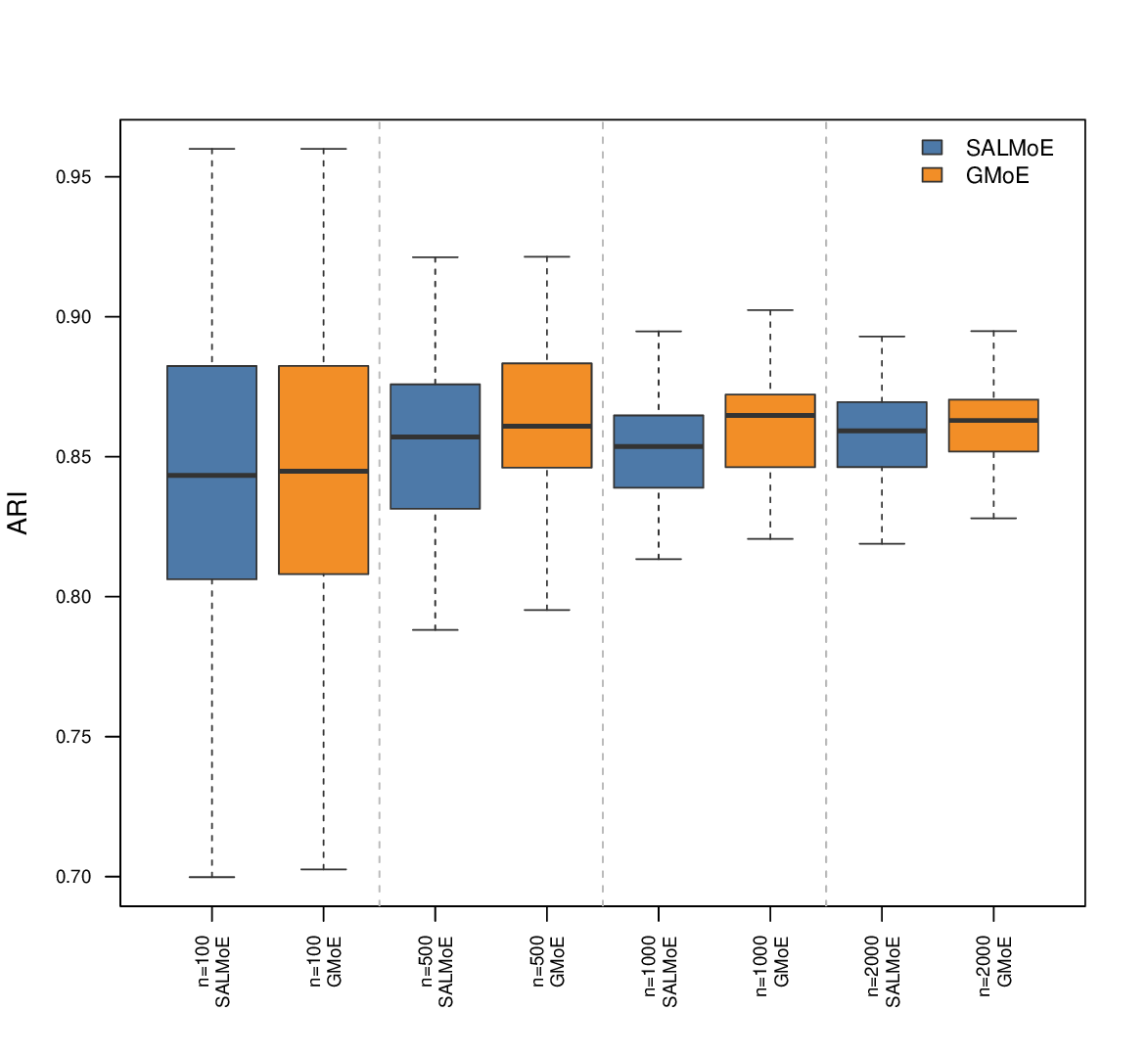}
        %\caption{Caption 1}
        \label{fig:fig1}
    \end{subfigure}
    \hfill
    \begin{subfigure}[b]{0.35\textwidth}
        \centering
        \includegraphics[width=\textwidth]{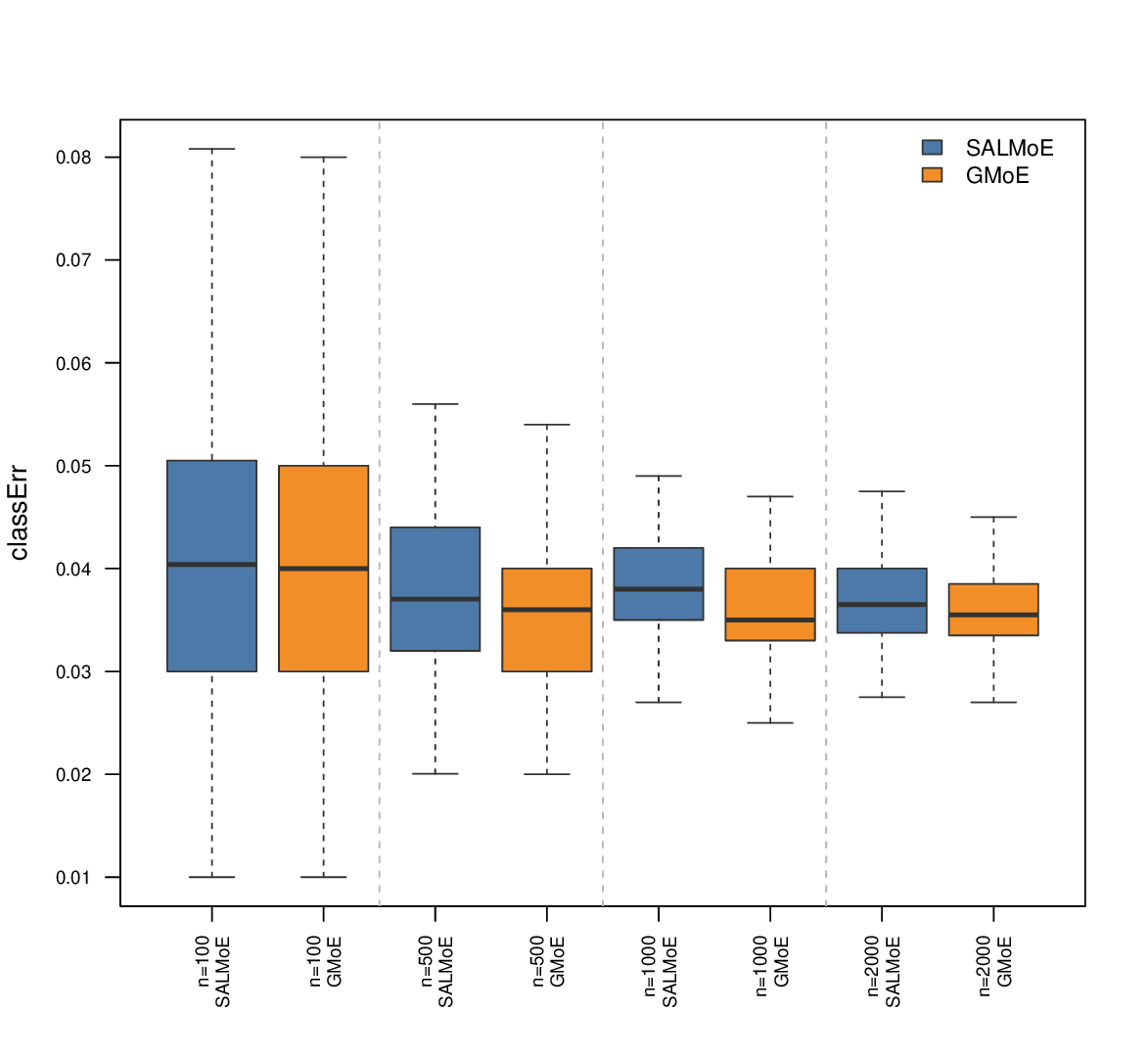}
        %\caption{Caption 2}
        \label{fig:fig2}
    \end{subfigure}

    %\vspace{0.2em}

    % Row 2
    \begin{subfigure}[b]{0.35\textwidth}
        \centering
        \includegraphics[width=\textwidth]{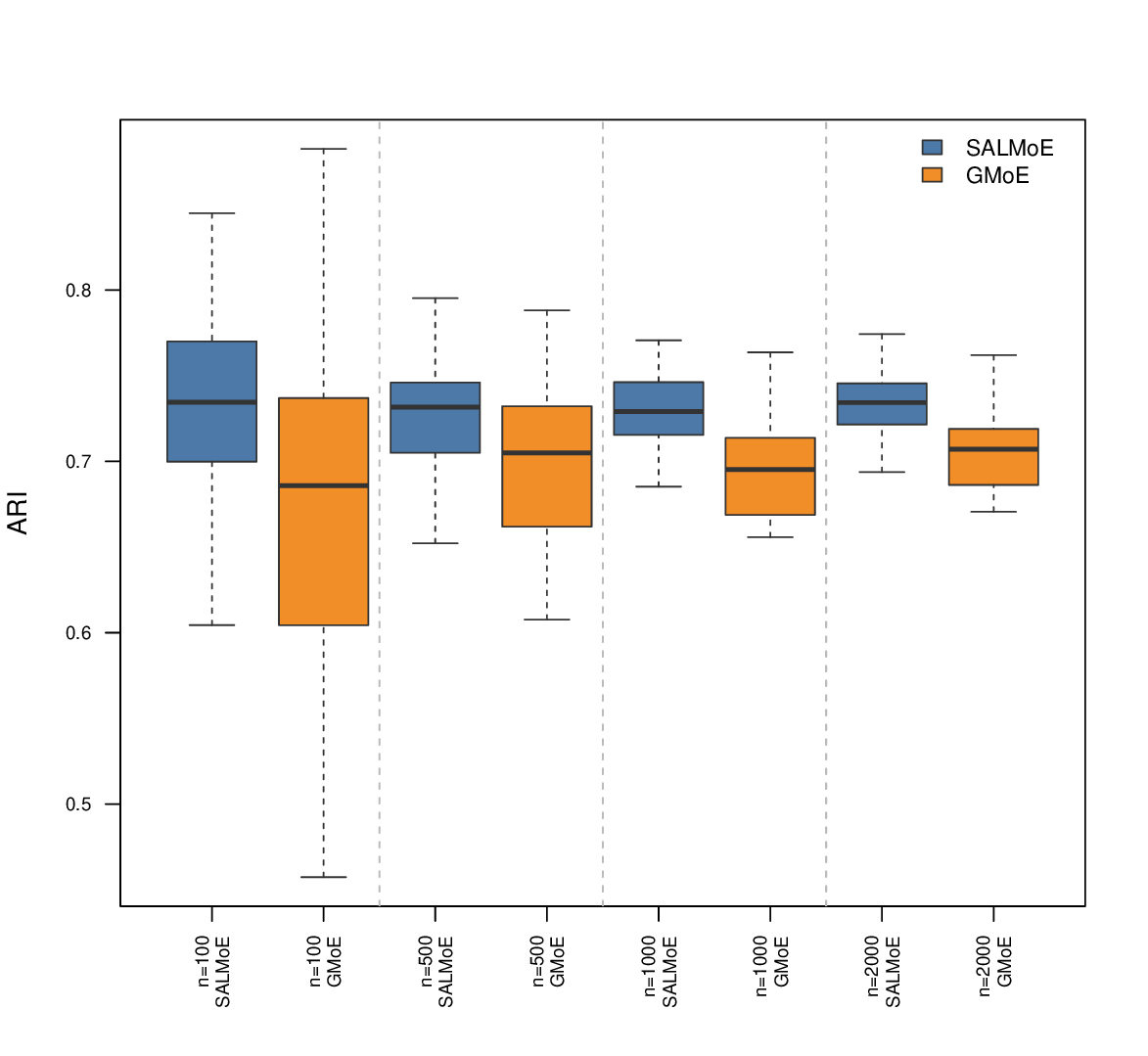}
        %\caption{Caption 3}
        \label{fig:fig3}
    \end{subfigure}
    \hfill
    \begin{subfigure}[b]{0.35\textwidth}
        \centering
        \includegraphics[width=\textwidth]{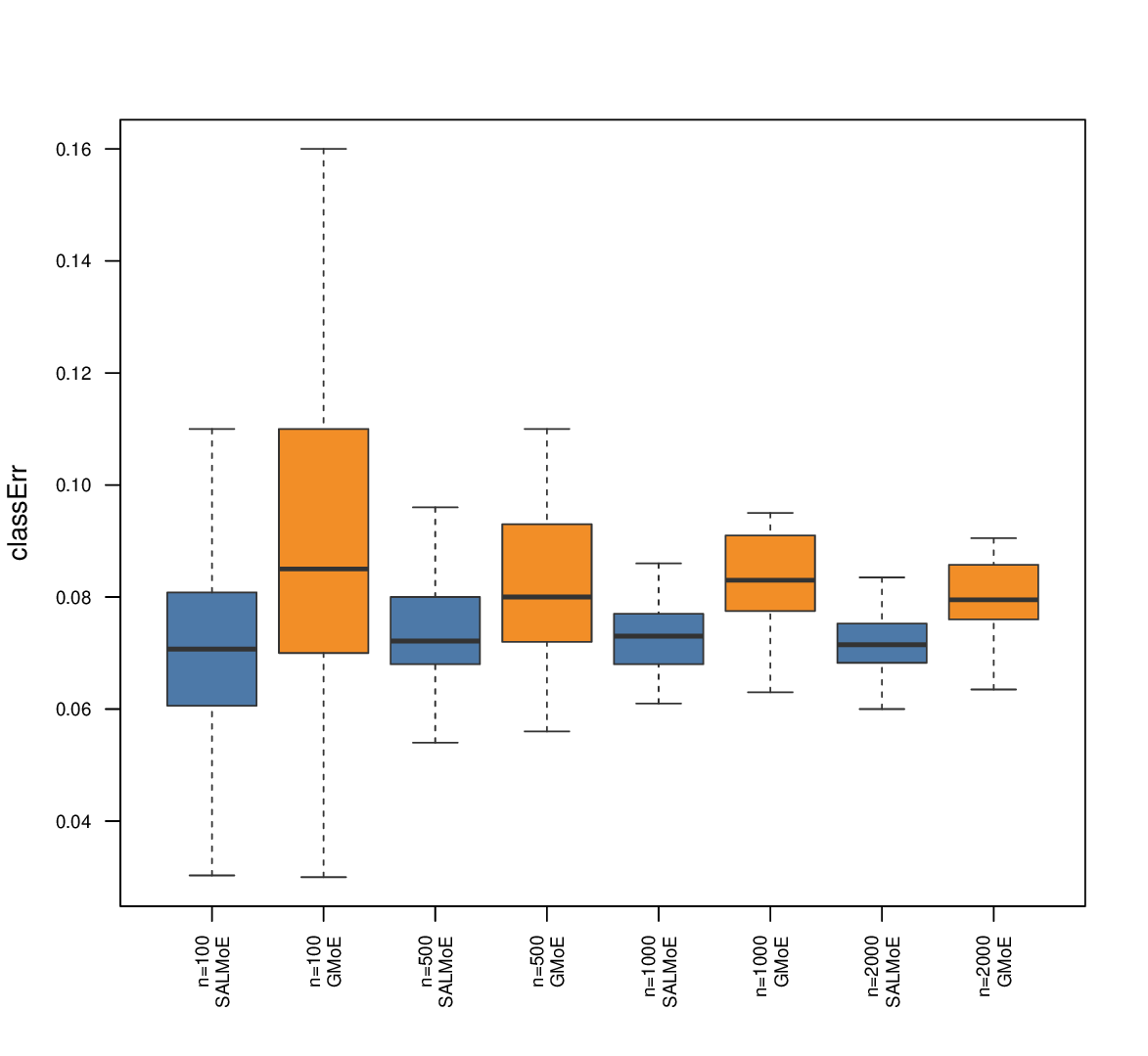}
        %\caption{Caption 4}
        \label{fig:fig4}
    \end{subfigure}

    %\vspace{0.2em}

    % Row 3
    \begin{subfigure}[b]{0.35\textwidth}
        \centering
        \includegraphics[width=\textwidth]{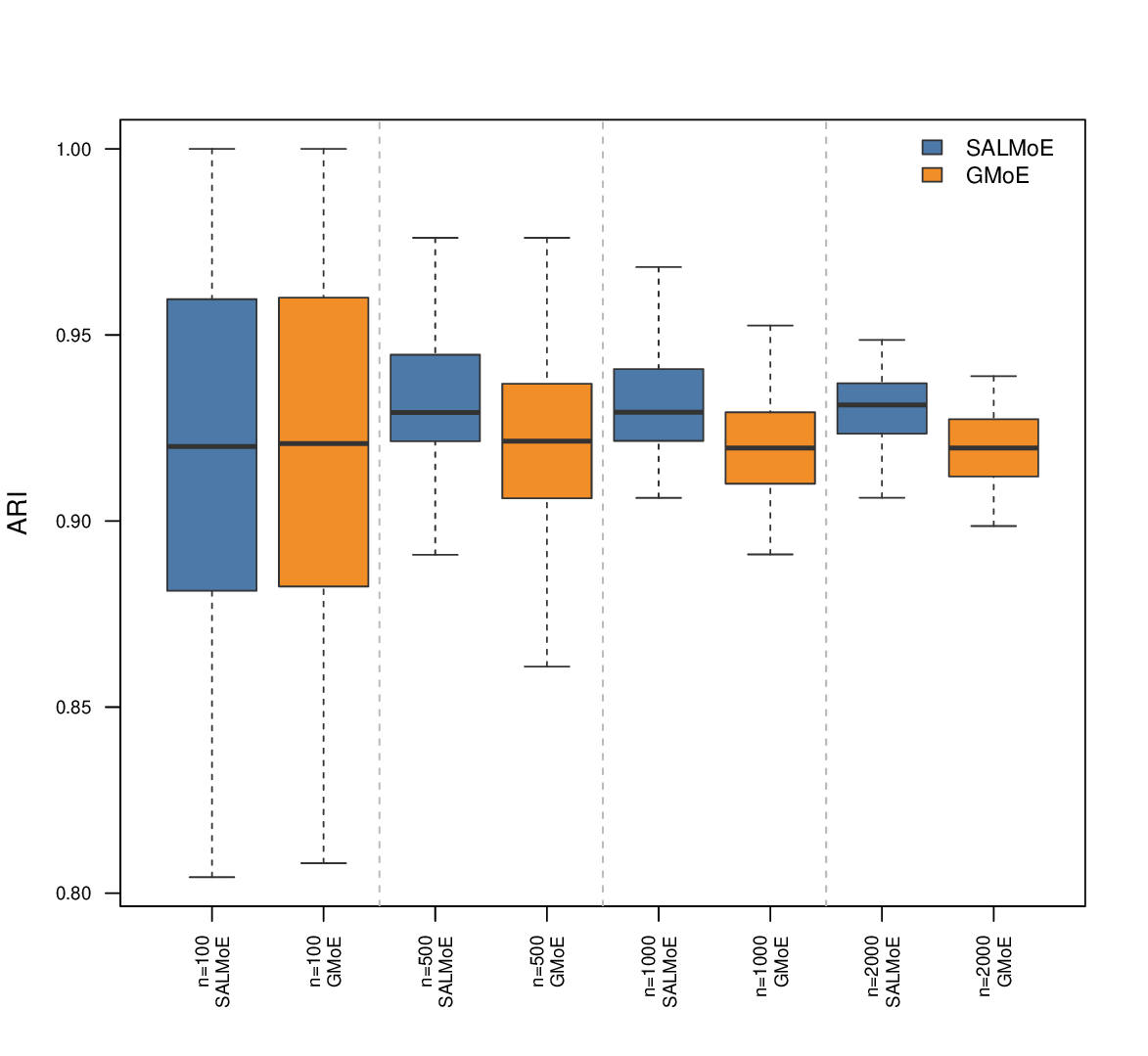}
        %\caption{Caption 5}
        \label{fig:fig5}
    \end{subfigure}
    \hfill
    \begin{subfigure}[b]{0.35\textwidth}
        \centering
        \includegraphics[width=\textwidth]{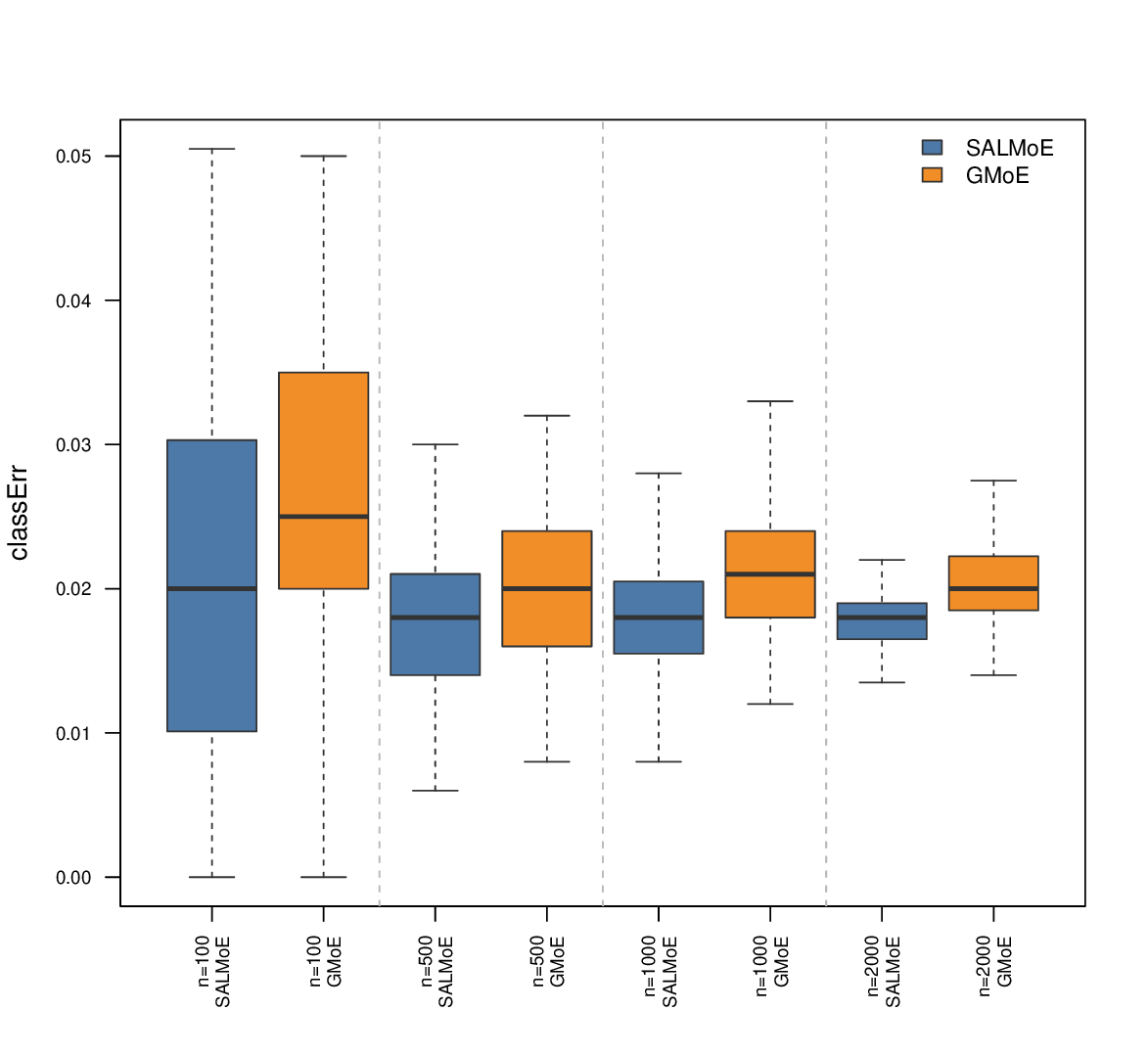}
        %\caption{Caption 6}
        \label{fig:fig6}
    \end{subfigure}

    %\vspace{0.2em}
    \begin{subfigure}[b]{0.35\textwidth}
        \centering
        \includegraphics[width=\textwidth]{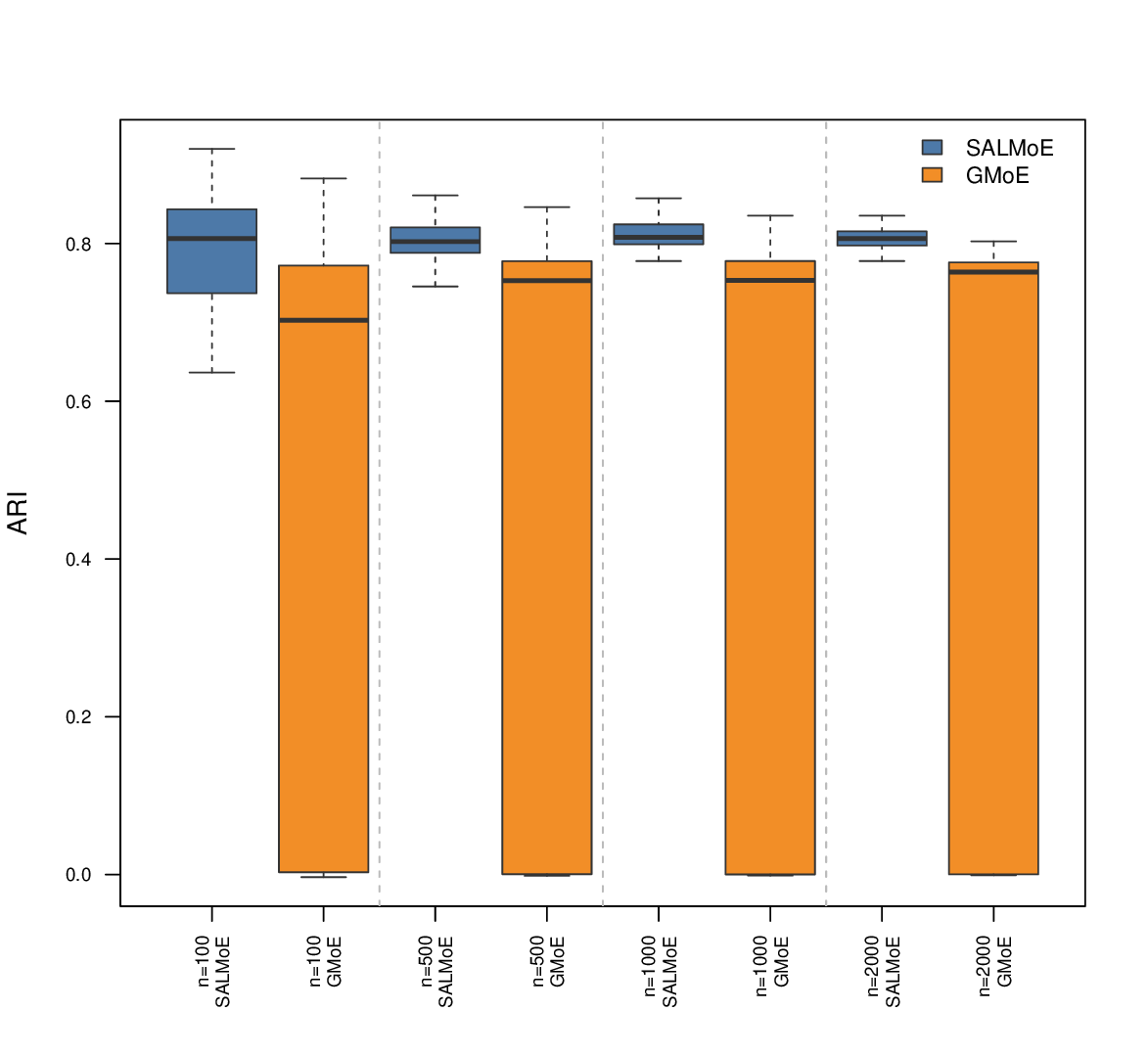}
        %\caption{Caption 7}
        \label{fig:fig7}
    \end{subfigure}
    \hfill
    \begin{subfigure}[b]{0.35\textwidth}
        \centering
        \includegraphics[width=\textwidth]{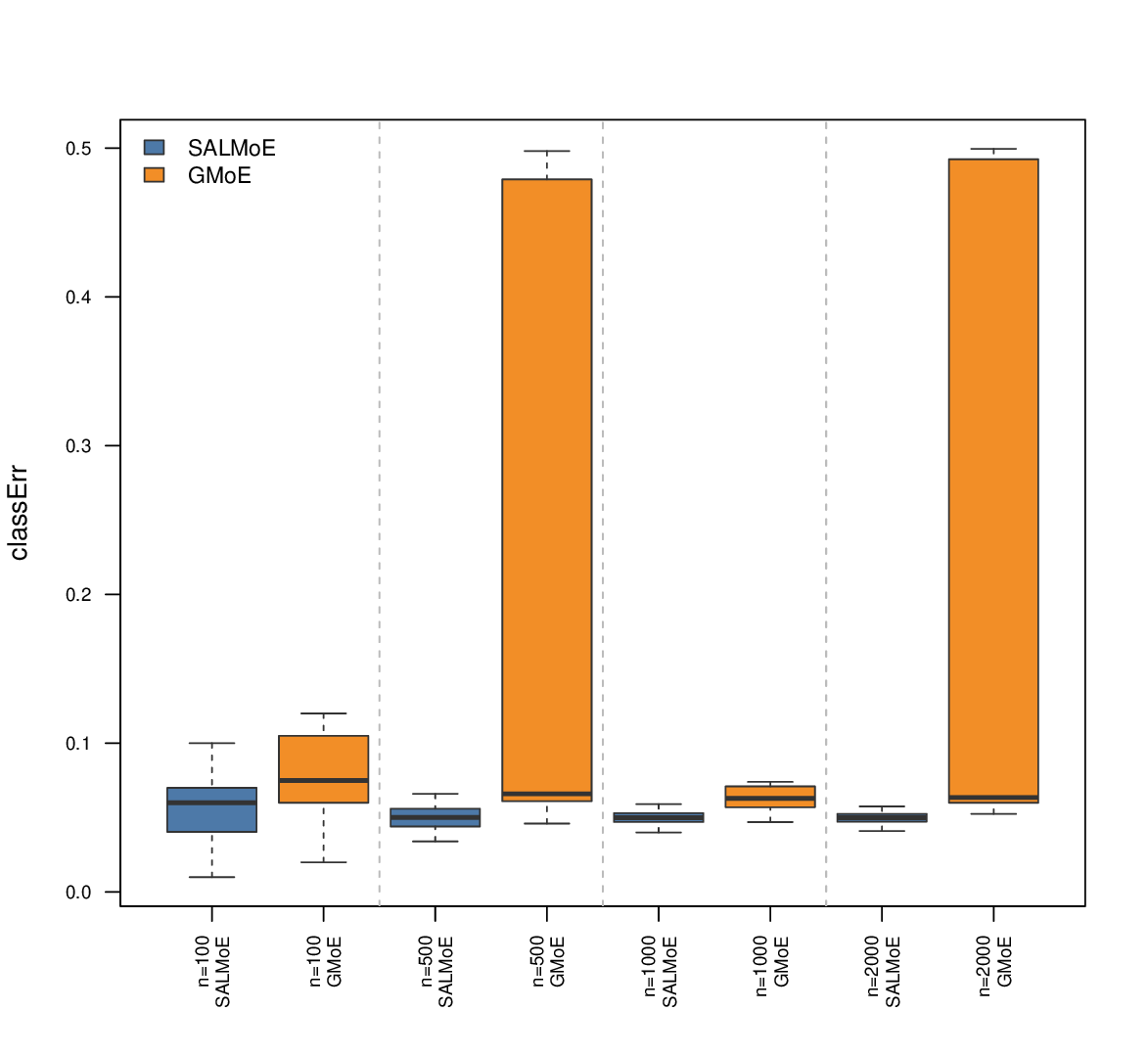}
        %\caption{Caption 8}
        \label{fig:fig8}
    \end{subfigure}
    \caption{Boxplot comparisons of the ARI (first column) and ClassErr (second column) calculated from the SALMoE (in blue) and GMoE (in orange) from Scenario (a) (first row) to Scenario (d) (last row).}
    \label{fig:clus}
\end{figure}
\subsection{Choosing the number of expert components}\label{choose_k}
To identify the suitable criterion for choosing the correct number of expert components for the proposed SALMoE model, we generate $100$ samples of size $n=500$ from an MoE model based on the following scenarios:
\begin{description}
    \item[S1:] $K=2$ Gaussian expert components with parameters given in Table \ref{tab:sim_params1}.
    \item[S2:] $K=2$ SAL expert components with parameters given in Table \ref{tab:sim_params1}.
    \item[S3:] $K=2$ skew-normal (SN) expert components with skewness parameters $\lambda_1=\lambda_2=20$ and the rest of the parameters are given in Table \ref{tab:sim_params1}.
    \item[S4:] $K=3$ Gaussian expert components with parameters $\boldsymbol{\eta}_1=\boldsymbol{\eta}_2=(0,10)$, $\boldsymbol\beta_1=(0,1)$, $\boldsymbol\beta_2=(0,-1)$, $\boldsymbol\beta_3=(1,-1)$ and $\sigma_1=\sigma_2=\sigma_3=0.1$.
    \item[S5:] $K=3$ SAL expert components with parameters $\boldsymbol{\eta}_1=\boldsymbol{\eta}_2=(0,10)$, $\boldsymbol\beta_1=(0,1)$, $\boldsymbol\beta_2=(0,-1)$, $\boldsymbol\beta_3=(1,-1)$, $\alpha_1=\alpha_2=\alpha_3=1$ and $\sigma_1=\sigma_2=\sigma_3=0.1$.
    \item[S6:] $K=3$ SN expert components with parameters $\boldsymbol{\eta}_1=\boldsymbol{\eta}_2=(0,10)$, $\boldsymbol\beta_1=(0,1)$, $\boldsymbol\beta_2=(0,-1)$, $\boldsymbol\beta_3=(1,-1)$, $\lambda_1=\lambda_2=\lambda_3=20$ and $\sigma_1=\sigma_2=\sigma_3=0.1$.
\end{description}
As already mentioned (see Section \ref{subsec4}), the value of $\alpha$ will be chosen using \eqref{pan_calibrate}. To ensure comparability between the PanIC, with orders $\beta\in\{1,2\}$, and the other ICs, we will use $\nu\in\{10^3,10^4\}$.\\
To each of the generated data, we fit a SALMoE model with $K=2,3,4$ and $5$ expert components. To evaluate the performance, we compute the average value of $K$, denoted by $\bar{K}$, obtained from each IC, as well as the number of times when $K=K^\star$, where $K^\star$ is the true value of $K$ used to generate the data. The results of this experiment are given in Table \ref{tab:sim_results}.\\
In Scenarios S1 and S3, all the ICs perform well and the results are comparable. In Scenario S2, when we have skewed and heavy-tailed data, the performance of the ICs is still good but not as good as in Scenarios S1 and S3. Moreover, the ICL performs slightly better than the other ICs (its $\bar{K}$ value is the closest one to the true value of $2$). This is also the case in Scenario S6. In Scenario S4, the ICL performs slightly better than the other ICs. In Scenario S5, when we have $K=3$ skewed and heavy-tailed expert components, the performance of all the ICs deteriorates (the ICs tend to choose the wrong number of expert components, i.e. $\#\{K=K^\star\}$ is small for all the ICs). However, the PanIC with $\alpha(1,10^3)$ has the closest $\bar{K}$ value to the true value of $3$, while the PanIC with $\alpha(2,10^4)$ gives the largest $\#\{K=K^\star\}$ value in this scenario.
\begin{table}[ht]
\centering
\caption{Results from 100 simulation runs for each scenario of the MoE models using sample size $n = 500$.}
\label{tab:sim_results}
\setlength{\tabcolsep}{4pt}
\renewcommand{\arraystretch}{0.9}
\small
\begin{tabular}{@{}lcl cccccc@{}}
\toprule
&&&&&\multicolumn{4}{c}{PanIC}\\
\cmidrule{6-9}
Scenario & $K^*$ & &BIC & ICL &
$\alpha(1, 10^3)$ & $\alpha(1, 10^4)$ &
$\alpha(2, 10^3)$ & $\alpha(2, 10^4)$ \\
\midrule
\multirow{2}{*}{S1}&\multirow{2}{*}{2}&$\bar{K}$&2&2&2&2.17&2&2.12 \\
&& \#\{$K$=$K^*$\}&100&100&100&89&100&91\\
\addlinespace
\multirow{2}{*}{S2}&\multirow{2}{*}{2}&$\bar{K}$&2.20&$\mathbf{2.17}$&2.20&2.31&2.20&2.27 \\
&&\#\{$K$=$K^*$\}&85&85&85&78&85&80 \\
\addlinespace
\multirow{2}{*}{S3}&\multirow{2}{*}{2}&$\bar{K}$&2&2&2&2.02&2&2.02\\
& & \#\{$K$=$K^*$\}&100&100&100&98&100&98\\
\addlinespace
\multirow{2}{*}{S4}&\multirow{2}{*}{3}&$\bar{K}$&3.40&$\mathbf{3.39}$&3.41&3.63&3.41&3.58\\
&&\#\{$K$=$K^*$\}&66&$\mathbf{67}$&66&54&66&56\\
\addlinespace
\multirow{2}{*}{S5}&\multirow{2}{*}{3} &$\bar{K}$&2.60&2.45&$\mathbf{2.86}$&3.44&2.83&3.38\\
&&\#\{$K$=$K^*$\}&32&32&38&41&37&$\mathbf{42}$\\
\addlinespace
\multirow{2}{*}{S6}&\multirow{2}{*}{3}&$\bar{K}$&3.16&$\mathbf{3.01}$&3.16&3.26&3.16&3.23\\
& & \#\{$K$=$K^*$\}&85&85&85&79&85&80\\
\bottomrule
\end{tabular}
 
\bigskip
\noindent\small Note: Bold values highlight that the corresponding IC has either the closest $\bar{K}$ to $K^*$ or the greatest \#\{$K$=$K^*$\} value for the indicated scenario.
\end{table}

\section{Application}\label{sec4}
In this section, we demonstrate the practical usefulness of the proposed methods on two real economic datasets.
\subsection{CO$_2$ data}
Our first dataset consists of CO$_2$ emissions per capita ($\mathrm{CO2pc}$) and GDP per capita ($\mathrm{GDPpc}$) for a group of 187 countries in 2024. The $\mathrm{CO2pc}$ data were obtained from \href{https://ourworldindata.org/co2-and-greenhouse-gas-emissions}{https://ourworldindata.org/co2-and-greenhouse-gas-emissions} and the $\mathrm{GDPpc}$ data were obtained from \href{https://ourworldindata.org/economic-growth}{https://ourworldindata.org/economic-growth}. Figure \ref{fig:hist_scatter_CO2data} shows a scatter plot of the data and histogram of $\mathrm{CO2pc}$. 
\begin{figure}[htb]
    \centering
    \begin{subfigure}[b]{0.45\textwidth}
        \centering
        \includegraphics[width=\textwidth]{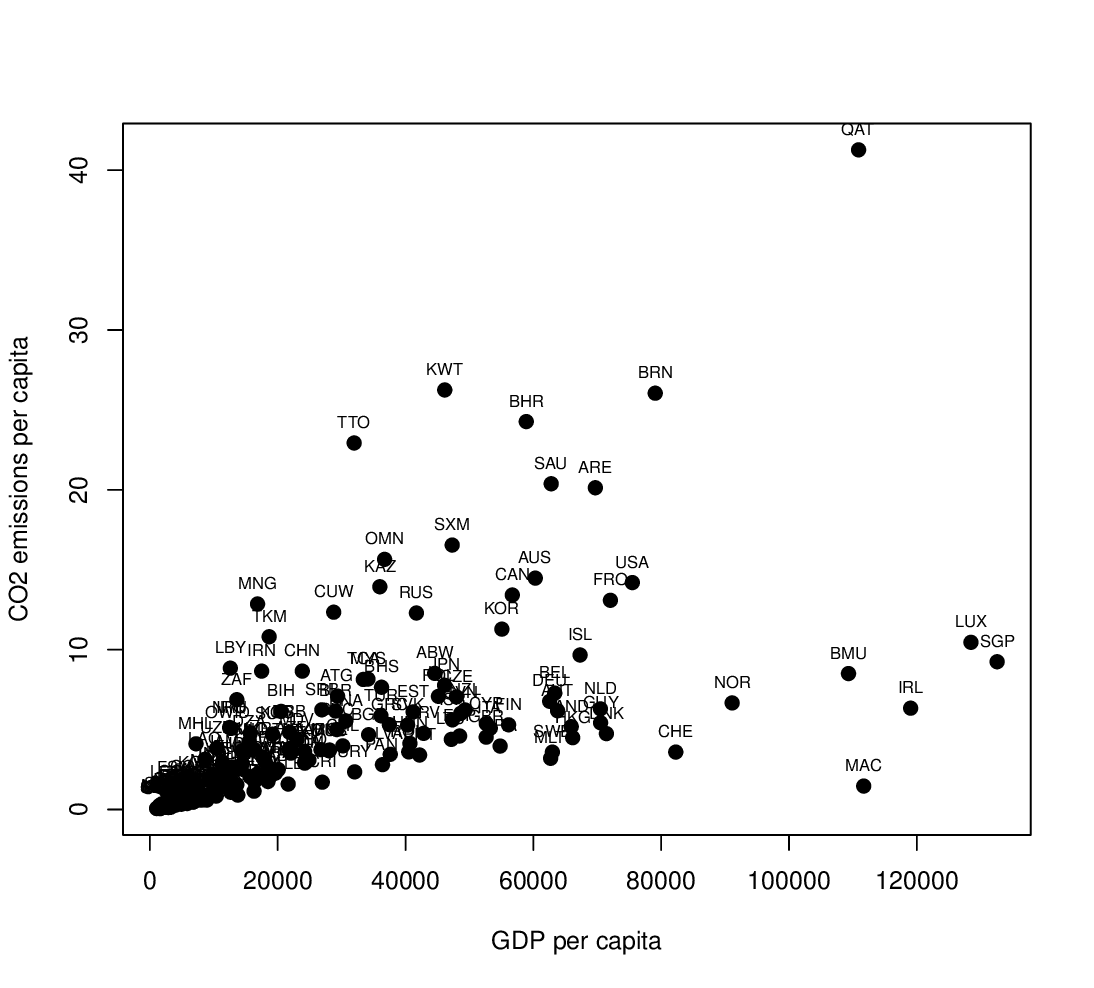}
        \caption{}
        \label{fig:co2_scatter}
    \end{subfigure}
    \hfill
    \begin{subfigure}[b]{0.45\textwidth}
        \centering
        \includegraphics[width=\textwidth]{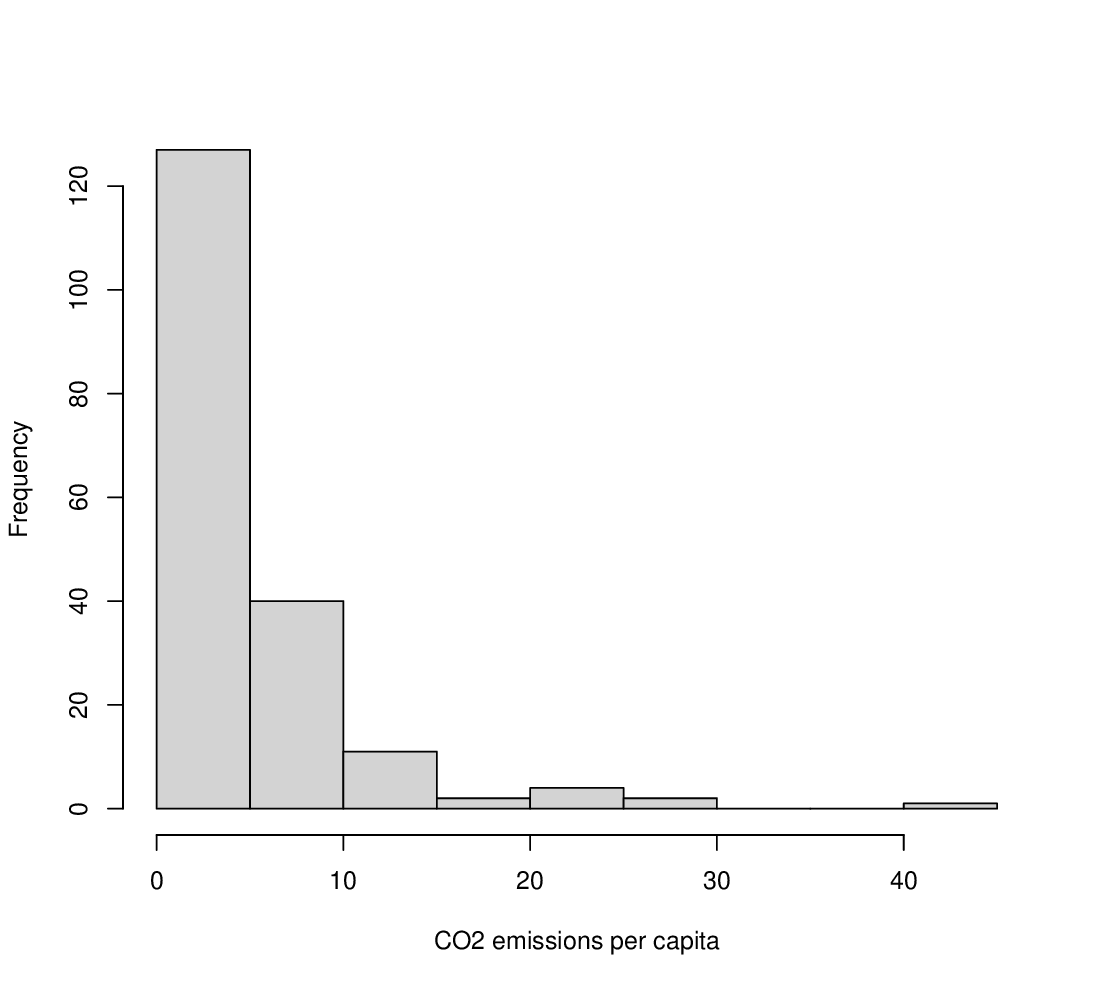}
        \caption{}
        \label{fig:co2_hist}
    \end{subfigure}
    \caption{CO$_2$ data: (a) Scatter plot of $\mathrm{CO2pc}$ and $\mathrm{GDPpc}$. (b) Histogram of $\mathrm{CO2pc}$}
    \label{fig:hist_scatter_CO2data}
\end{figure}
It can be seen from the histogram that the distribution of $\mathrm{CO2pc}$ is skewed. We are interested in the dependence of $\mathrm{CO2pc}$ on $\mathrm{GDPpc}$. To investigate the presence of multiple latent regression relationships (experts) in the data, we fit the SALMoE model with $K=1,2,\dots,5$. We compare the SALMoE model with the GMoE model and choose the best model based on the ICs. The variables are standardized before fitting the models.\\ 
Let $\mathbf{x}=\mathbf{t}=(1,\mathrm{GDPpc})$ be the design vector for the expert and gating functions. Table \ref{CO2data_IC} gives the IC values of all the fitted models. It can be seen from the table that the $K=2$ expert component SALMoE model gives the best fit for the data based on all the ICs. In Figure \ref{fig:fitted_mean}, we plot the obtained component mean functions and the overall mean function of the best model. We also show the data points that are assigned to each component, obtained using the MAP approach in Section \ref{subsec6}.  Table \ref{CO2data_param} gives the estimated parameters of the best model. Also included in the table are the 95\% Bootstrap confidence intervals obtained in a similar manner as in \cite{Skhosana2026}.

\begin{table}[ht]
\centering
\caption{CO$_2$ data: Information criteria values for the estimated SALMoE model and GMoE model for $K=1,2,\dots,5$.}\label{CO2data_IC}
\setlength{\tabcolsep}{9pt}
\renewcommand{\arraystretch}{1.2}
\begin{tabular}{cccccccc}
        \hline
        \multirow{2}{*}{$K$}& \multicolumn{3}{c}{SALMoE} && \multicolumn{3}{c}{GMoE} \\
        \cmidrule{2-4}\cmidrule{5-8}
        & BIC & ICL & PanIC && BIC& ICL&PanIC \\
        \hline
        1&272.4465&272.4465&259.8493&&304.7228&304.7228&295.2749 \\
        \hline
        2 &$\mathbf{220.1228}$&$\mathbf{249.5628}$ &$\mathbf{188.6298}$ &&306.0764&306.038&280.843 \\
        \hline
        3 &243.3229&295.5258 &192.9341&&314.5188&314.5196&273.578 \\
        \hline
        4 &269.0602&335.1393&199.7756 &&330.7039&330.6989&274.0165 \\
        \hline
        5 &294.3384&397.1369&206.158&& 358.7872&358.8232&286.3533 \\
        \hline
    \end{tabular}
\par\smallskip\noindent\small Note: The PanIC was calculated using the value of $\alpha=\alpha(\beta,\nu)$, where $\beta=1$ and $\nu=1000$.
\end{table}

\begin{figure}[htbp]
    \centering
    \begin{subfigure}[b]{0.45\textwidth}
        \centering
        \includegraphics[width=\textwidth]{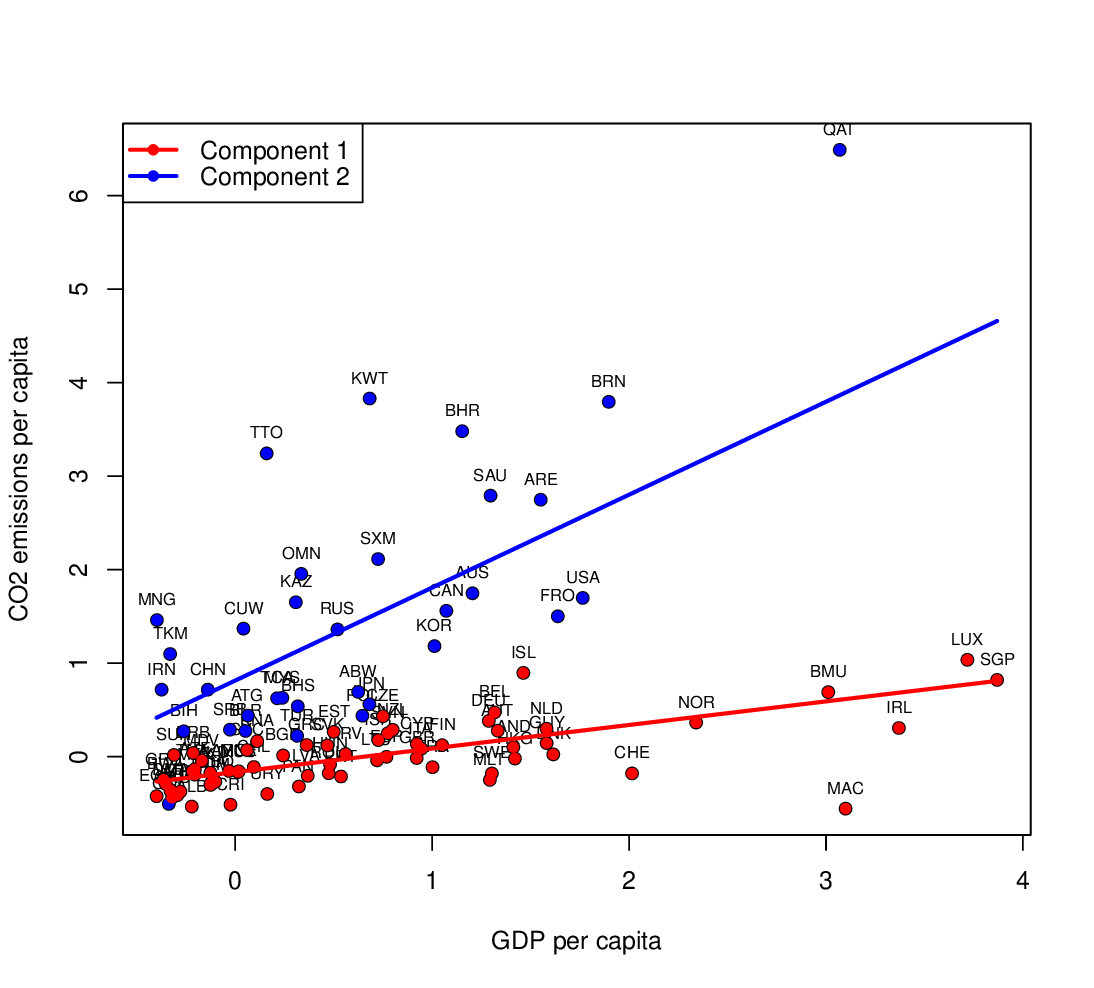}
        \caption{}
        \label{fig:co2_components}
    \end{subfigure}
    \hfill
    \begin{subfigure}[b]{0.45\textwidth}
        \centering
        \includegraphics[width=\textwidth]{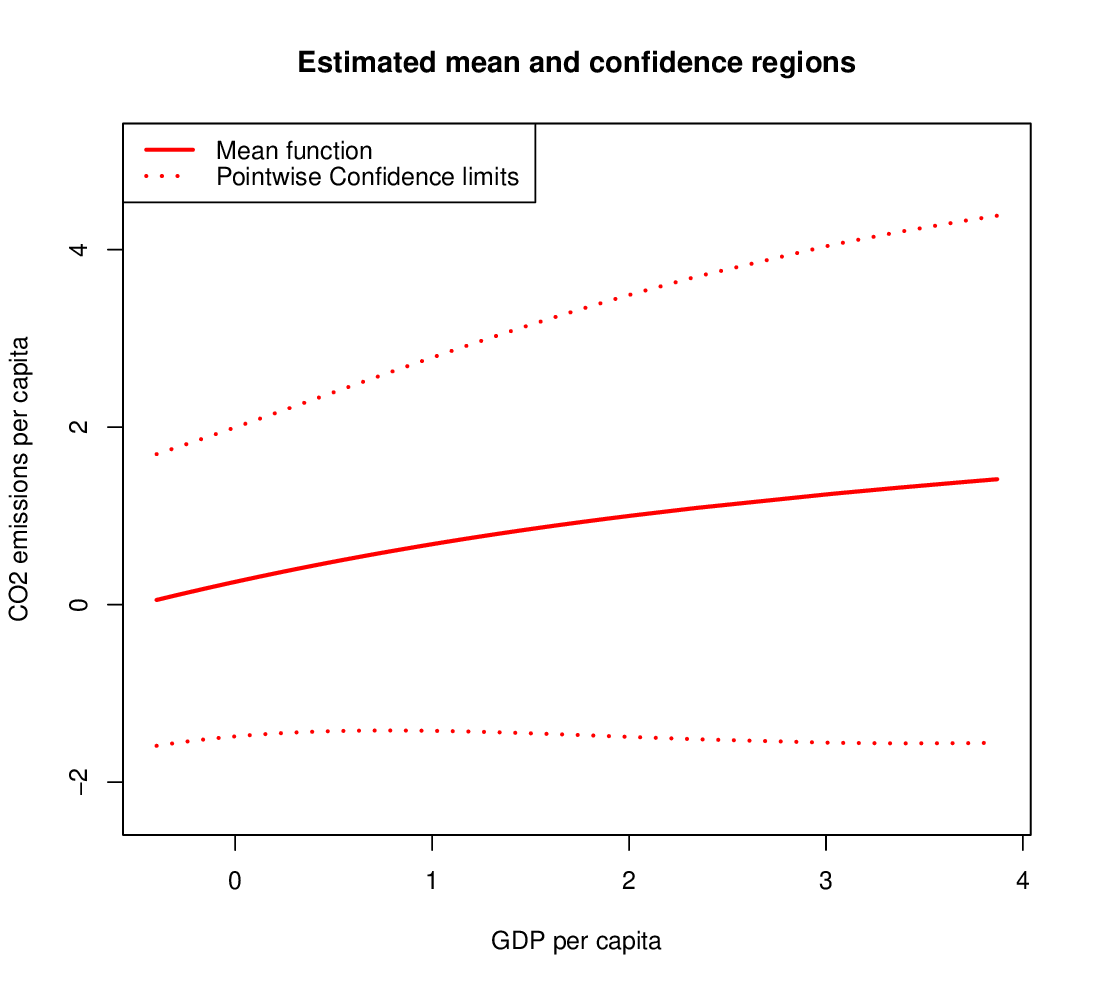}
        \caption{}
        \label{fig:co2_overall_mean}
    \end{subfigure}
    \caption{CO$_2$ data: (a) Fitted component regression functions and (b) the overall mean function (with the pointwise predictive limits)}
    \label{fig:fitted_mean}
\end{figure}
\begin{table}[ht]
\centering
\caption{CO$_2$ data: Parameter estimates of the fitted $K=2$ component SALMoE model and the $95\%$ Bootstrap confidence intervals}\label{CO2data_param}
\setlength{\tabcolsep}{6pt}
\renewcommand{\arraystretch}{1.2}
\begin{tabular}{lccccccccccc}
\hline
&$\beta_{10}$&$\beta_{11}$&$\beta_{20}$&$\beta_{21}$&$\sigma_1$&$\sigma_2$&$\alpha_1$&$\alpha_2$&$\eta_{10}$&$\eta_{11}$ \\
\hline
&$-0.1542$&$0.2517$&$-0.1762$&$1.0233$&$0.0864$&$0.1100$&$-0.008$&$0.9887$&$0.3086$&$0.3533$\\
\hline
95\% CI Lower&-0.2106&0.2057&-0.2310&0.9458&0.0425&0.0141&-0.1145&0.7403&-0.2147&0.0026\\
95\% CI Upper&-0.0924&0.2987&-0.0933&1.1440&0.1318&0.2336&0.1083&1.2525&0.8097&0.8302\\
\hline
\end{tabular}
\end{table}
For countries in component 2, the increase in $\mathrm{CO2pc}$, following an increase in $\mathrm{GDPpc}$, is greater than that in component 1. Thus, countries in component 2 have a high emission profile compared to countries in component 1. This may reflect greater reliance on carbon-emitting fossil fuels among countries in this fitted component. This group includes the United States (USA) and China (CHN), which have high emission profiles under the fitted model. In contrast, component 1 includes countries such as Norway (NOR) and Sweden (SWE), which appear in the lower-emission component. Figure \ref{fig:mixing_comp1} plots the fitted mixing proportion function for component 1. It can be seen from the figure that as GDP per capita increases, a given country is more likely to belong to component 1. This can be interpreted as saying that at higher levels of growth per capita (increased $\mathrm{GDPpc}$), countries move towards information-intensive economies that prioritise sustainability by leveraging technological advances and enforcing environmental regulations \cite{dinda2004}. Incidentally, these results are in line with the well-known environmental Kuznets curve (EKC) hypothesis (see \cite{dinda2004} for more details).\\
Finally, Figure \ref{fig:loglik_EM} shows the observed log-likelihood values at each iteration of the hybrid EM-MM algorithm for fitting the $K=2$ SALMoE model.
\begin{figure}[htbp]
    \centering
    \begin{subfigure}[b]{0.45\textwidth}
        \centering
        \includegraphics[width=\textwidth]{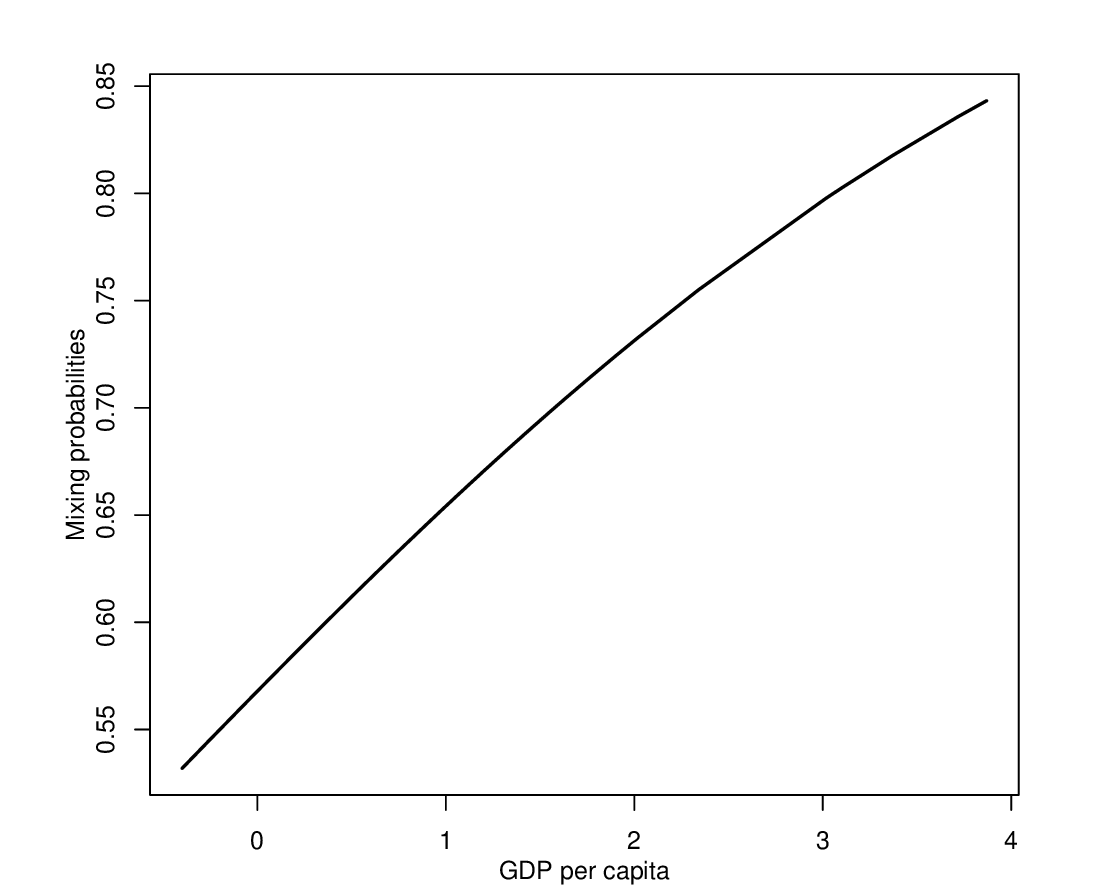}
        \caption{}
        \label{fig:mixing_comp1}
    \end{subfigure}
    \hfill
    \begin{subfigure}[b]{0.45\textwidth}
        \centering
        \includegraphics[width=\textwidth]{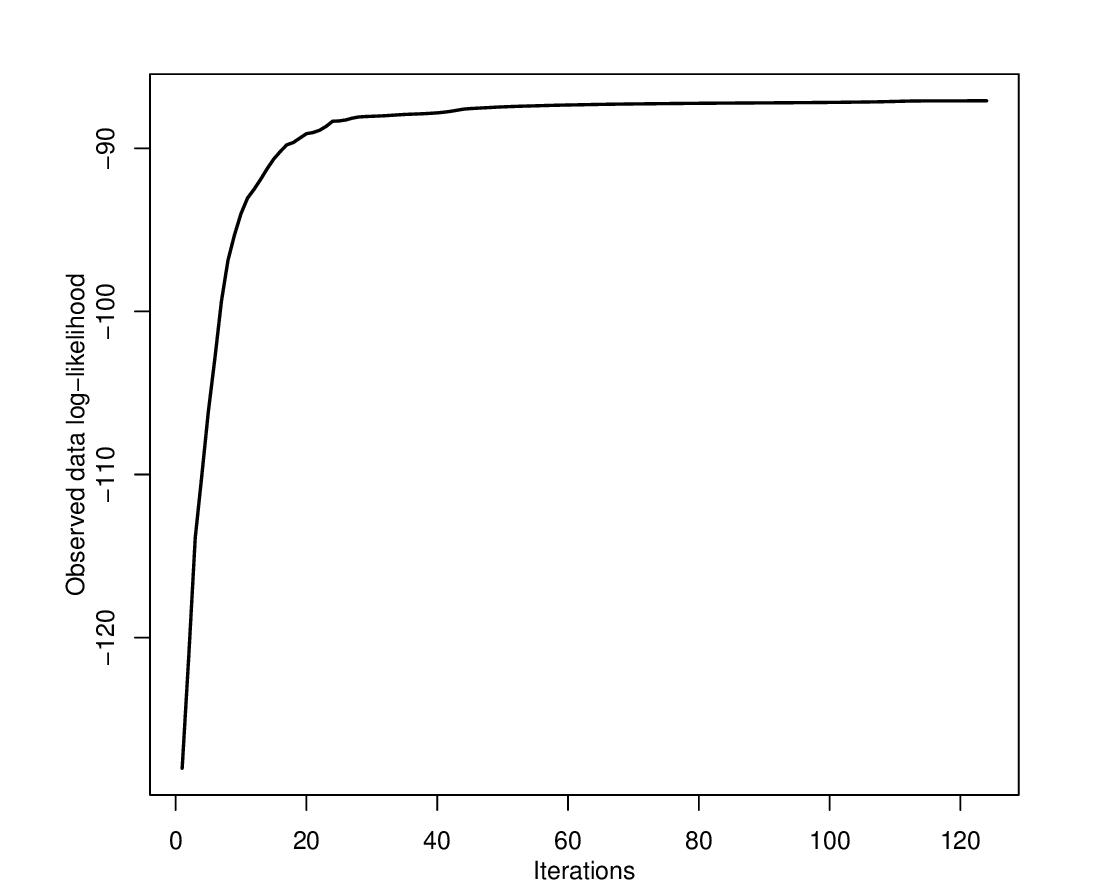}
        \caption{}
        \label{fig:loglik_EM}
    \end{subfigure}
    \caption{CO$_2$ data: (a) Fitted mixing probability function for component 1 and (b) the observed data log-likelihood values at each iteration of the hybrid EM-MM algorithm.}
    \label{fig:fitted_gating}
\end{figure}
\subsection{GDP growth data}
The second dataset consists of economic data from 88 countries averaged over a five-year period 1960--1964. The data were extracted from the cross-country GDP growth panel data for the period 1960--1995 which can be found in the \texttt{R} package \texttt{np} \cite{hayfield2008}. The data comprise the growth rate of real gross domestic product (GDP) per capita ($\mathrm{growth}$), per capita real GDP at the beginning of the five-year period ($\mathrm{initgdp}$), average annual population growth for the five-year period ($\mathrm{popgro}$), average investment-to-GDP ratio for the five-year period ($\mathrm{inv}$), average enrolment rate in secondary schools for the five-year period ($\mathrm{humancap}$) and a categorical variable ($\mathrm{oecd}$) with $1$ for OECD countries and $0$ otherwise.\\
We are interested in the dependence of $\mathrm{growth}$ on $\mathrm{initgdp}$, $\mathrm{popgro}$, $\mathrm{inv}$ and $\mathrm{humancap}$.

\begin{figure}[htbp]
	\centering
	\includegraphics[width=\textwidth]{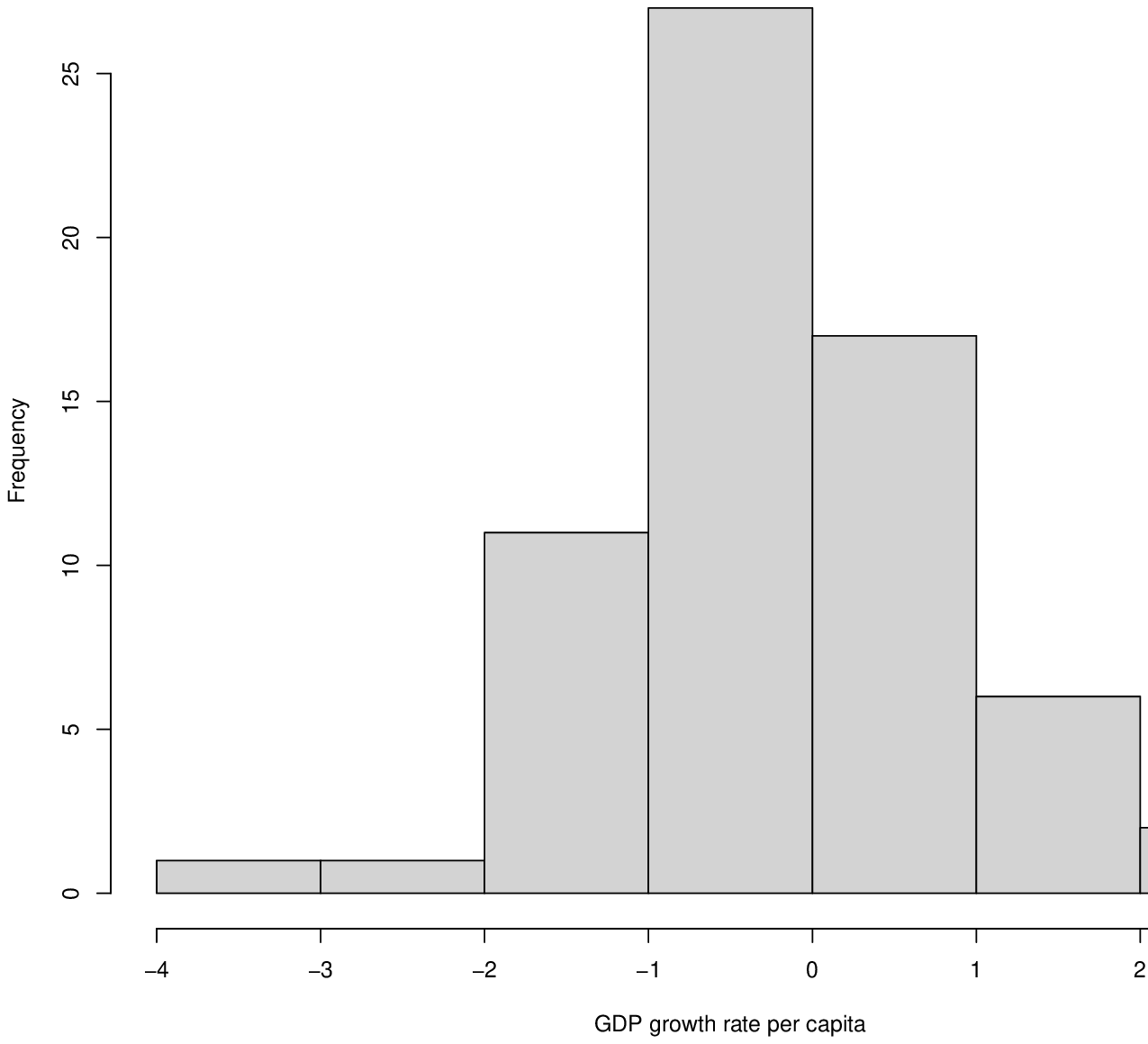}
	\caption{GDP growth data: Histogram of per capita GDP growth rate for non-OECD countries (left-panel) and OECD countries (right-panel)}
	\label{fig:growth_hist}
\end{figure}

%\begin{figure}[htbp]
	%\centering
	%\begin{subfigure}[b]{0.45\textwidth}
	%	\centering
	%	\includegraphics[height=5cm]{OECD}
	%	\caption{}
	%	\label{fig:growth_hist_OECD}
	%\end{subfigure}
	%\hfill
	%\begin{subfigure}[b]{0.45\textwidth}
	%	\centering
%		\includegraphics[height=5cm]{nonOECD}
%		\caption{}
%		\label{fig:growth_hist_non_OECD}
%	\end{subfigure}
%    \caption{GDP growth data: Histogram of per capita GDP growth rate for (a) OECD countries and (b) non-OECD countries}
%\label{fig:growth_hist}
%\end{figure}

Figure \ref{fig:growth_hist} plots the histograms of the $\mathrm{growth}$ variable for non-OECD countries (left-panel) and OECD countries (right-panel). Notice that the distribution of $\mathrm{growth}$ is different between OECD and non-OECD countries. Thus, we will use the $\mathrm{oecd}$ variable as an observed reference partition for the two groups. We will then fit the $K=2$ component SALMoE and GMoE models to the combined data to see whether we can recover the two groups. We standardize all the variables before fitting the models.\\
Let $\mathbf{x}=\mathbf{t}=(1,\mathrm{initgdp},\mathrm{popgro},\mathrm{inv},\mathrm{humancap})$ be the design vector of the experts and gating functions.
Table \ref{Growthdata_IC} gives the IC values calculated for the $K=2$ component SALMoE model and GMoE model. It can be seen from the table that the SALMoE model is the best $K=2$ component model for this dataset. Table \ref{tab:Growthdata_param} gives the estimated parameters of the best model. It can be seen that the impact of the covariates on the response variable differs widely between the two fitted components. For instance, for non-OECD countries (component 1), the influence of schooling (i.e. $\mathrm{humancap}$) has a positive and large impact on the growth in GDP of a country. In contrast, for OECD countries (component 2), $\mathrm{humancap}$ has a negligible impact on the growth in GDP of a country. Moreover, this impact is statistically insignificant (the CI includes a zero).

\begin{table}[hb]
\centering
\caption{GDP growth data: Information criteria values for the estimated SALMoE model and GMoE model for $K=2$.}\label{Growthdata_IC}
\setlength{\tabcolsep}{9pt}
\renewcommand{\arraystretch}{1.2}
\begin{tabular}{cccccccc}
    \hline
    \multirow{2}{*}{}& \multicolumn{3}{c}{SALMoE} && \multicolumn{3}{c}{GMoE} \\
    \cmidrule{2-4}\cmidrule{5-8}
    & BIC & ICL & PanIC && BIC& ICL&PanIC \\
    \hline
    &257.5521&271.6111&197.7184&&274.9852&289.6982&221.4498 \\
    \hline
    \end{tabular}
\par\smallskip\noindent\small Note: The PanIC was calculated using the value of $\alpha=\alpha(\beta,\nu)$, where $\beta=1$ and $\nu=1000$.
\end{table}

\begin{table}[htbp]
\caption{Estimated model parameters for the GDP growth rate data based on the best model (i.e. $K=2$ SALMoE model) and the 95\% Bootstrap confidence intervals}
\label{tab:Growthdata_param}
\setlength{\tabcolsep}{9pt}
\renewcommand{\arraystretch}{1.2}
\centering
\begin{tabular}{||c|c|c|c||c|c|c|c||}
\hline
Parameter & Estimate&\multicolumn{2}{c}{95\% CI} & Parameter & Estimate&\multicolumn{2}{c}{95\% CI}\\
\hline
&&Lower&Upper&&&Lower&Upper\\
\hline
$\beta_{10}$ & $-0.4962$&$-0.6993$&$-0.2902$& $\sigma_1$    & $0.0258$&$0.0002$&$0.0428$  \\
$\beta_{11}$ & $-1.7841$&$-1.9326$ &$-1.6220$&$\sigma_2$    & $0.6018$&$0.2728$&$0.8679$  \\
$\beta_{12}$ & $-0.2304$&$-0.3086$&$-0.1406$& $\alpha_1$    & $-0.0512$&$-0.1822$&$0.0766$ \\
$\beta_{13}$ & $ 1.0742$&$0.9469$&$1.2338$& $\alpha_2$    & $0.3073$ &$0.0161$&$0.5905$ \\
$\beta_{14}$ & $ 1.8729$&$1.5393$&$2.1548$& $\eta_{10}$   & $-2.4290$&$-4.7303$&$-1.4591$ \\
$\beta_{20}$ & $-0.3361$&$-0.5478$ &$-0.0992$&  $\eta_{11}$   & $-0.8305$&$-2.7791$&$0.3456$ \\
$\beta_{21}$ & $ 0.4393$&$0.1214$&$0.7279$& $\eta_{12}$   & $-0.2287$&$-1.4652$&$0.6518$ \\
$\beta_{22}$ & $-0.1582$&$-0.3531$&$0.0355$& $\eta_{13}$   & $-1.6261$&$-3.9352$&$-0.1125$ \\
 $\beta_{23}$ & $-0.0694$&$-0.3032$&$0.1621$ & $\eta_{14}$     & $4.0168$&$1.9662$&$8.4177$  \\
$\beta_{24}$ & $-0.0440$ &$-0.3020$&$0.2683$&&&&\\
\hline
\end{tabular}
\end{table}
Let $z_{oecd}$ be an indicator variable taking a value of $1$ if a country is an OECD country and $0$ otherwise. Let $\hat z_{oecd}$ be the estimated label of a given country obtained using the MAP approach in Section \ref{subsec6}. We calculated the classification accuracy, $n^{-1}\sum^n_{i=1}\mathbb{I}(z_{i,oecd}=\hat{z}_{i,oecd})$, and found that about $74\%$ of the countries were classified correctly as OECD or non-OECD countries; equivalently, the normalised class error is about $26\%$. Finally, Figure \ref{fig:Log_lik_Growthdata} shows the log-likelihood values at each iteration of the hybrid EM-MM algorithm. It can be seen that the log-likelihood values are increasing at each iteration.
\begin{figure}[htbp]
    \centering
    \includegraphics[width=0.7\textwidth]{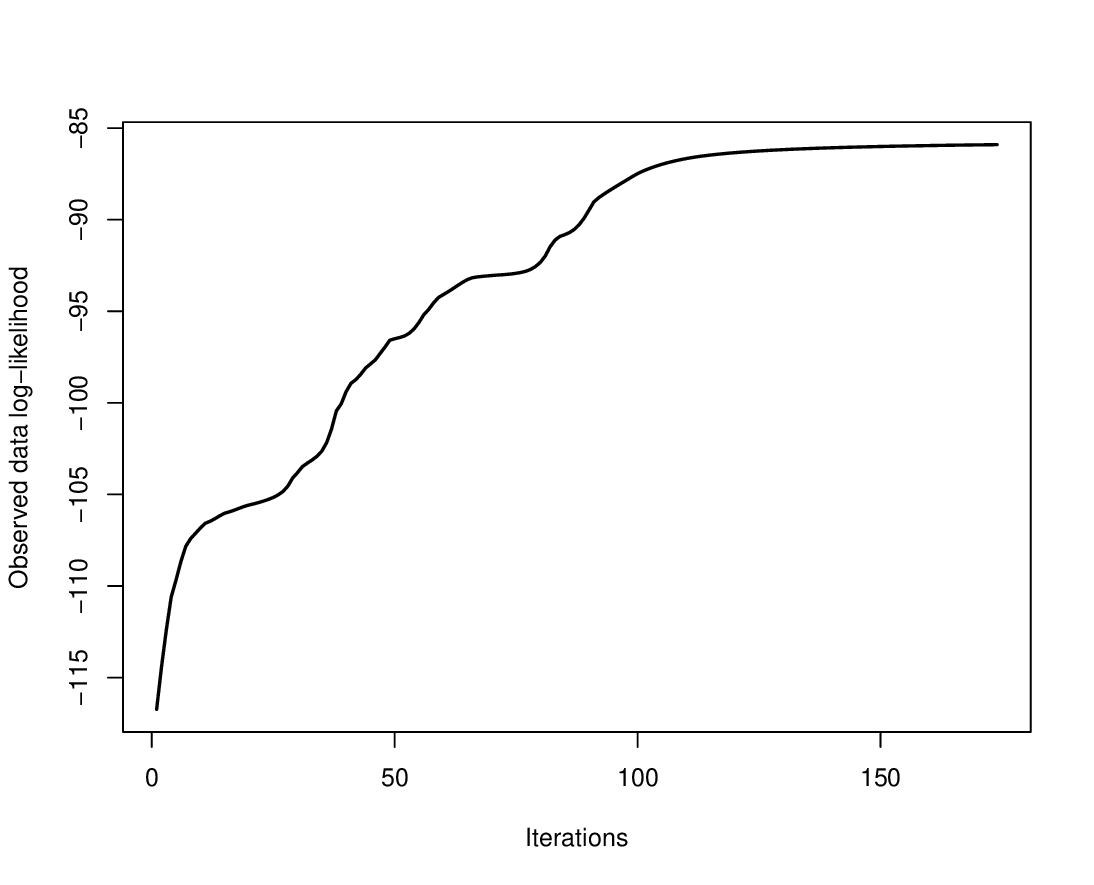}
    \caption{GDP growth data: The observed data log-likelihood values at each iteration of the hybrid EM-MM algorithm.}
    \label{fig:Log_lik_Growthdata}
\end{figure}
\section{Conclusion}\label{sec13}\label{sec5}
This paper proposes a new mixture of experts model for non-Gaussian data (i.e. skewed and possibly heavy-tailed data with outliers) using the shifted asymmetric Laplace distribution for the expert distributions. The model is termed the SALMoE model. To fit the model, we develop a hybrid algorithm combining the classical Expectation-Maximization (EM) algorithm and the minorization-maximization (MM) algorithm for estimating the parameters of the gating functions. A theoretical analysis of the convergence of the hybrid EM-MM algorithm was conducted. An intensive simulation study showed the robustness of the proposed SALMoE model under different scenarios involving skewed and heavy-tailed data. To choose the number of expert components for the SALMoE model, we considered the BIC, ICL and PanIC. Based on a numerical study that considered a variety of scenarios, the three approaches provided comparable results. To demonstrate the practical usefulness of the proposed methods, we applied them to two real datasets.\\
In this paper, we assumed that the gating functions and the component regression functions are parametric functions of the covariates. However, in practice, the parametric assumption may not hold. A natural future extension of this work is to consider nonparametric gating functions and component regression functions. Another possible extension of the current work is to consider the multivariate response variable case.
\backmatter

\appendix

\section{Appendix}
\subsection{Derivation of the MM update for the gating parameters}\label{MMstep}
For $h=1,\dots,K-1$, a standard calculation gives 
\begin{equation}
\frac{\partial}{\partial\boldsymbol{\eta}_{h}}Q(\tilde{\boldsymbol{\eta}}\mid\boldsymbol{\eta}^{(old)})=\sum_{i=1}^{n}\left\{ \gamma_{ih}^{(new)}-\pi_{h}(\mathbf{t}_{i}\mid\boldsymbol{\eta})\right\} \mathbf{t}_{i}.\label{grad_blocks}
\end{equation}
Stacking the blocks yields the score vector 
\begin{equation}
\boldsymbol{g}(\tilde{\boldsymbol{\eta}})=\frac{\partial}{\partial\tilde{\boldsymbol{\eta}}}Q(\tilde{\boldsymbol{\eta}}\mid\boldsymbol{\eta}^{(old)})=\sum_{i=1}^{n}\left\{ \boldsymbol{\gamma}_{i}-\mathbf{p}_{i}(\tilde{\boldsymbol{\eta}})\right\} \otimes\mathbf{t}_{i}.\label{score_vec}
\end{equation}
Moreover, the Hessian has the Kronecker form 
\begin{equation}
\mathbf{H}(\tilde{\boldsymbol{\eta}})=\frac{\partial^{2}}{\partial\tilde{\boldsymbol{\eta}}\,\partial\tilde{\boldsymbol{\eta}}^{\top}}Q(\tilde{\boldsymbol{\eta}}\mid\boldsymbol{\eta}^{(old)})=-\sum_{i=1}^{n}\left\{ \mathbf{A}_{\mathbf{p}_{i}(\tilde{\boldsymbol{\eta}})}-\mathbf{p}_{i}(\tilde{\boldsymbol{\eta}})\mathbf{p}_{i}(\tilde{\boldsymbol{\eta}})^{\top}\right\} \otimes(\mathbf{t}_{i}\mathbf{t}_{i}^{\top}),\label{hessian_form}
\end{equation}
where $\mathbf{A}_{\mathbf{p}}=\operatorname{diag}(\mathbf{p})$.

\noindent Let $\mathbf{1}\in\mathbb{R}^{K-1}$ be the vector of ones
and define the symmetric matrix 
\begin{equation}
\mathbf{C}=\mathbf{I}_{K-1}-\frac{1}{K}\mathbf{1}\mathbf{1}^{\top}.\label{Cmat_def}
\end{equation}
The multinomial logistic regression bound of \citet{Bohning1992MultinomialLogitAlgorithm}
implies that, for every $\mathbf{p}\in(0,1)^{K-1}$ arising from a
$K$-category multinomial logit model with a baseline category, 
\begin{equation}
\mathbf{A}_{\mathbf{p}}-\mathbf{p}\mathbf{p}^{\top}\preceq\frac{1}{2}\mathbf{C},\label{bohning_cov_bound}
\end{equation}
in the Loewner ordering. Combining (\ref{hessian_form}) and (\ref{bohning_cov_bound})
and using the property that $\mathbf{M}\preceq\mathbf{N}$ and $\mathbf{S}\succeq\mathbf{0}$
imply $\mathbf{M}\otimes\mathbf{S}\preceq\mathbf{N}\otimes\mathbf{S}$,
we obtain the global Hessian bound 
\begin{equation}
\mathbf{H}(\tilde{\boldsymbol{\eta}})\succeq\mathbf{B}\quad\text{for all }\tilde{\boldsymbol{\eta}},\qquad\mathbf{B}=-\frac{1}{2}\mathbf{C}\otimes(\mathbf{T}^{\top}\mathbf{T}).\label{hessian_lower_bound}
\end{equation}
Assuming $\mathbf{T}^{\top}\mathbf{T}$ is nonsingular (full column
rank gating design), $\mathbf{B}$ is negative definite.

We now construct a proper MM step around the current gating iterate
$\tilde{\boldsymbol{\eta}}^{(old)}$. By a second-order Taylor expansion,
for each $\tilde{\boldsymbol{\eta}}$ there exists $a\in(0,1)$ such
that 
\begin{align}
Q(\tilde{\boldsymbol{\eta}}\mid\boldsymbol{\eta}^{(old)}) & =Q(\tilde{\boldsymbol{\eta}}^{(old)}\mid\boldsymbol{\eta}^{(old)})+(\tilde{\boldsymbol{\eta}}-\tilde{\boldsymbol{\eta}}^{(old)})^{\top}\boldsymbol{g}(\tilde{\boldsymbol{\eta}}^{(old)})\nonumber \\
 & \quad+\frac{1}{2}(\tilde{\boldsymbol{\eta}}-\tilde{\boldsymbol{\eta}}^{(old)})^{\top}\mathbf{H}\!\left(\tilde{\boldsymbol{\eta}}^{(old)}+a(\tilde{\boldsymbol{\eta}}-\tilde{\boldsymbol{\eta}}^{(old)})\right)(\tilde{\boldsymbol{\eta}}-\tilde{\boldsymbol{\eta}}^{(old)}).\label{taylor_eta_one}
\end{align}
Applying the global bound (\ref{hessian_lower_bound}) inside (\ref{taylor_eta_one})
yields the quadratic minoriser 
\begin{eqnarray}
Q(\tilde{\boldsymbol{\eta}}\mid\boldsymbol{\eta}^{(old)})\ge\mathcal{M}(\tilde{\boldsymbol{\eta}}\mid\tilde{\boldsymbol{\eta}}^{(old)})=Q(\tilde{\boldsymbol{\eta}}^{(old)}\mid\boldsymbol{\eta}^{(old)})+(\tilde{\boldsymbol{\eta}}-\tilde{\boldsymbol{\eta}}^{(old)})^{\top}\boldsymbol{g}(\tilde{\boldsymbol{\eta}}^{(old)})\nonumber\\
+\frac{1}{2}(\tilde{\boldsymbol{\eta}}-\tilde{\boldsymbol{\eta}}^{(old)})^{\top}\mathbf{B}(\tilde{\boldsymbol{\eta}}-\tilde{\boldsymbol{\eta}}^{(old)})\nonumber,\label{minorizer_def_one}
\end{eqnarray}
with equality at $\tilde{\boldsymbol{\eta}}=\tilde{\boldsymbol{\eta}}^{(old)}$.
Since $\mathbf{B}$ is negative definite, $\mathcal{M}(\cdot\mid\tilde{\boldsymbol{\eta}}^{(old)})$
is strictly concave and has the unique maximiser 
\begin{equation}
\tilde{\boldsymbol{\eta}}^{(new)}=\arg\max_{\tilde{\boldsymbol{\eta}}}\,\mathcal{M}(\tilde{\boldsymbol{\eta}}\mid\tilde{\boldsymbol{\eta}}^{(old)})=\tilde{\boldsymbol{\eta}}^{(old)}-\mathbf{B}^{-1}\boldsymbol{g}(\tilde{\boldsymbol{\eta}}^{(old)}).\label{LB_update_one}
\end{equation}
Equivalently, using $\mathbf{B}^{-1}=-2\,\mathbf{C}^{-1}\otimes(\mathbf{T}^{\top}\mathbf{T})^{-1}$,
\begin{equation}
\tilde{\boldsymbol{\eta}}^{(new)}=\tilde{\boldsymbol{\eta}}^{(old)}+2\left\{ \mathbf{C}^{-1}\otimes(\mathbf{T}^{\top}\mathbf{T})^{-1}\right\} \boldsymbol{g}(\tilde{\boldsymbol{\eta}}^{(old)}).\label{LB_update_explicit_one}
\end{equation}
A convenient closed form for $\mathbf{C}^{-1}$ is 
\begin{equation}
\mathbf{C}^{-1}=\left(\mathbf{I}_{K-1}-\frac{1}{K}\mathbf{1}\mathbf{1}^{\top}\right)^{-1}=\mathbf{I}_{K-1}+\mathbf{1}\mathbf{1}^{\top},\label{C_inverse_one}
\end{equation}
which follows by direct multiplication. Writing the expression \eqref{LB_update_explicit_one} in matrix form gives rise to the one-step update equation \eqref{LB_update_matrix_one} for $\boldsymbol{\eta}$.

\subsection{Convergence of the EM-MM algorithm}
Let $\ell(\boldsymbol{\vartheta})$ denote the observed-data log-likelihood, and set $f(\boldsymbol{\vartheta})=-\ell(\boldsymbol{\vartheta})$,
where $\boldsymbol{\vartheta}=(\boldsymbol{\theta},\boldsymbol{\eta})$
with $\boldsymbol{\theta}=\{(\boldsymbol{\beta}_{k},\alpha_{k},\sigma_{k})\}_{k=1}^{K}$
and $\boldsymbol{\vartheta}\in\Theta$, with $\Theta$ as in the standing conventions.
We say that $\boldsymbol{\vartheta}^{\star}$ is a stationary point of the minimisation problem for $f$ on $\Theta$ if the directional derivative $f'(\boldsymbol{\vartheta}^{\star};\mathbf{d})$ is nonnegative for every feasible direction $\mathbf{d}$ at $\boldsymbol{\vartheta}^{\star}$.

At the current iterate $\boldsymbol{\vartheta}^{(m)}$, the E-step
defines the standard EM surrogate
\[
Q(\boldsymbol{\vartheta}\mid\boldsymbol{\vartheta}^{(m)})=Q_{\boldsymbol{\theta}}(\boldsymbol{\theta}\mid\boldsymbol{\vartheta}^{(m)})+Q_{\boldsymbol{\eta}}(\boldsymbol{\eta}\mid\boldsymbol{\vartheta}^{(m)}),
\]
where $Q_{\boldsymbol{\theta}}$ and $Q_{\boldsymbol{\eta}}$ are the expert and gating contributions, respectively. By the usual EM Jensen argument (see \citet[Section~8.5]{RazaviyaynHongLuo2013BSUM}),
there exists an upper-bound function 
\begin{equation}
 u_{m}(\boldsymbol{\vartheta})=-Q(\boldsymbol{\vartheta}\mid\boldsymbol{\vartheta}^{(m)})+\text{const}(\boldsymbol{\vartheta}^{(m)})\label{u_def_conv_revised}
\end{equation}
such that $f(\boldsymbol{\vartheta})\le u_{m}(\boldsymbol{\vartheta})$
for all $\boldsymbol{\vartheta}\in\Theta$ and $f(\boldsymbol{\vartheta}^{(m)})=u_{m}(\boldsymbol{\vartheta}^{(m)})$.
Thus, any update that decreases $u_{m}$ yields a valid MM (more precisely,
GEM) descent step for $f$.

Let $\mathcal{M}_{m}(\boldsymbol{\eta})=\mathcal{M}(\boldsymbol{\eta}\mid\boldsymbol{\eta}^{(m)})$ denote the quadratic minoriser of $Q_{\boldsymbol{\eta}}(\boldsymbol{\eta}\mid\boldsymbol{\vartheta}^{(m)})$ obtained in Appendix \ref{MMstep}. Then, for all $\boldsymbol{\eta}\in\Theta_{\boldsymbol{\eta}}$,
\[
\mathcal{M}_{m}(\boldsymbol{\eta})\le Q_{\boldsymbol{\eta}}(\boldsymbol{\eta}\mid\boldsymbol{\vartheta}^{(m)}),
\]
with equality and first-order agreement at $\boldsymbol{\eta}=\boldsymbol{\eta}^{(m)}$. Since replacing $-Q_{\boldsymbol{\eta}}$ by $-\mathcal{M}_{m}$ gives an upper bound for the gating part of $f$, define
\begin{align}
\tilde{u}_{m}(\boldsymbol{\theta},\boldsymbol{\eta})
&=u_{m}(\boldsymbol{\theta},\boldsymbol{\eta})+\left\{Q_{\boldsymbol{\eta}}(\boldsymbol{\eta}\mid\boldsymbol{\vartheta}^{(m)})-\mathcal{M}_{m}(\boldsymbol{\eta})\right\} \nonumber\\
&=-Q_{\boldsymbol{\theta}}(\boldsymbol{\theta}\mid\boldsymbol{\vartheta}^{(m)})-\mathcal{M}_{m}(\boldsymbol{\eta})+\text{const}(\boldsymbol{\vartheta}^{(m)}).\label{utilde_def_revised}
\end{align}
Because $Q_{\boldsymbol{\eta}}(\boldsymbol{\eta}\mid\boldsymbol{\vartheta}^{(m)})-\mathcal{M}_{m}(\boldsymbol{\eta})\ge0$, we have $\tilde{u}_{m}(\boldsymbol{\vartheta})\ge u_{m}(\boldsymbol{\vartheta})\ge f(\boldsymbol{\vartheta})$ for all $\boldsymbol{\vartheta}\in\Theta$. Moreover, by the tightness of the EM surrogate and the minoriser,
$\tilde{u}_{m}(\boldsymbol{\vartheta}^{(m)})=f(\boldsymbol{\vartheta}^{(m)})$.

The EM-MM update can now be written as the minimisation of the separable upper bound \eqref{utilde_def_revised}. In particular,
\begin{equation}
\boldsymbol{\theta}^{(m+1)}\in\arg\min_{\boldsymbol{\theta}\in\Theta_{\boldsymbol{\theta}}}\{-Q_{\boldsymbol{\theta}}(\boldsymbol{\theta}\mid\boldsymbol{\vartheta}^{(m)})\},\label{theta_update_conv_revised}
\end{equation}
and
\begin{equation}
\boldsymbol{\eta}^{(m+1)}\in\arg\min_{\boldsymbol{\eta}\in\Theta_{\boldsymbol{\eta}}}\{-\mathcal{M}_{m}(\boldsymbol{\eta})\}.\label{eta_update_conv_revised}
\end{equation}
The first update is exactly M-step~1. The second update is equivalent to maximising the quadratic minoriser $\mathcal{M}_{m}$ and therefore gives the one-step update \eqref{LB_update_one}. Consequently,
\[
f(\boldsymbol{\vartheta}^{(m+1)})\le \tilde{u}_{m}(\boldsymbol{\vartheta}^{(m+1)})\le \tilde{u}_{m}(\boldsymbol{\vartheta}^{(m)})=f(\boldsymbol{\vartheta}^{(m)}),
\]
or equivalently, $\ell(\boldsymbol{\vartheta}^{(m)})$ is nondecreasing.

We next record the standard conditions under which the above monotone MM scheme has stationary limit points. The upper bound \eqref{utilde_def_revised} satisfies the usual tightness, global upper-bound, first-order agreement and continuity conditions used in the BSUM framework of \citet{RazaviyaynHongLuo2013BSUM}. Indeed, tightness and the global upper-bound property follow from the construction above; first-order agreement follows from the first-order agreement of the EM surrogate together with the first-order tightness of $\mathcal{M}_{m}$ at $\boldsymbol{\eta}^{(m)}$; and continuity follows from the continuity of the E-step expectations and the quadratic form of $\mathcal{M}_{m}$.

To apply the convergence results of \citet{RazaviyaynHongLuo2013BSUM}, we assume the usual additional conditions: (i) the level set $\{\boldsymbol{\vartheta}\in\Theta_c:f(\boldsymbol{\vartheta})\le f(\boldsymbol{\vartheta}^{(0)})\}$ is compact in the compact parameter set $\Theta_c$ defined in the standing conventions, (ii) $f$ is regular on this level set in the sense of \citet{RazaviyaynHongLuo2013BSUM}, and (iii) the relevant block subproblems have unique minimisers, or ties are resolved by a closed deterministic rule. The regularity assumption is not automatic for nonsmooth nonconvex objectives, so we impose it explicitly. It is nevertheless mild in the present setting because $f$ is directionally differentiable and piecewise smooth, with possible kinks arising only on residual-zero hyperplanes. Under this regularity assumption, coordinate-wise stationarity of the block updates implies stationarity in the directional sense defined above. Uniqueness holds for the gating block whenever $\mathbf{T}^{\top}\mathbf{T}$ is nonsingular, since then the quadratic bound in M-step~2 is strictly convex, and holds for the expert block under the corresponding weighted nonsingularity conditions in M-step~1.

Under these conditions, the BSUM convergence results of \citet{RazaviyaynHongLuo2013BSUM} imply that every limit point of the outer EM-MM iterates $\{\boldsymbol{\vartheta}^{(m)}\}$ is a stationary point of the observed-data maximum likelihood problem. Equivalently, if $S^{\star}$ denotes the set of stationary points of $f$, then
\[
\lim_{m\to\infty}d(\boldsymbol{\vartheta}^{(m)},S^{\star})=0.
\]

\subsection{Some details about the implementation of the PanIC}\label{PanIC}
In the likelihood setting, the PanIC can be applied by taking the per-observation loss to be the negative log-likelihood, and using a penalty that separates models by a slowly diverging factor. A particularly simple choice is the Sin--White information criterion (SWIC) \cite{sin1996}. The SWIC penalty is defined as follows
\[
P_{K,n}=\alpha\text{df}(K)\frac{\log_{+}^{(\beta)}(n)}{\sqrt{n}},\qquad\alpha>0,
\]
where 
\[
\log_{+}^{(\beta)}(n)=\underbrace{\log_{+}\circ\log_{+}\circ\cdots\circ\log_{+}}_{\beta\ \text{times}}(n),\qquad\beta\in\mathbb{N},
\]
with $\log_{+}(x)=\max\{1,\log(x)\}$.\\
The corresponding SWIC (on the $-2\ell$ scale) is 
\begin{eqnarray}
\text{SWIC}_{\beta}(K)&=&-2\ell(\hat{\boldsymbol{\vartheta}}_{K})+2nP_{K,n}\nonumber\\
&=&-2\ell(\hat{\boldsymbol{\vartheta}}_{K})+2\alpha\text{df}(K)\sqrt{n}\log_{+}^{(\beta)}(n)\nonumber\\
&=& \text{PanIC}(K).
\end{eqnarray}
Thus, the SWIC is a member of the class of PanICs. With the calibration in \eqref{pan_calibrate}, the penalty term equals $\text{df}(K)\log(\nu)$ when $n=\nu$, so the criterion is on the same scale as the BIC at the reference sample size.\\
Note that the PanIC theory requires, for each fixed $K$, that the parameter space is compact and that the per-observation loss is Carath\'eodory and Lipschitz in the parameters with an $L^{2}$ envelope. In our setting, these conditions are verified on the compact parameter set $\Theta_c$ defined in the standing conventions, together with mild moment conditions on $(Y,\mathbf{X},\mathbf{T})$. Under these conditions, the SAL expert contribution is piecewise smooth in the regression parameters (the only kinks occur on residual-zero hyperplanes through absolute-value terms), the gating contribution is smooth, and hence the negative log-likelihood is Lipschitz in $\boldsymbol{\vartheta}$ with a square-integrable Lipschitz envelope. Therefore the assumptions of \citet{Nguyen2024PanIC} apply, and the resulting SWIC selector is consistent for the parsimonious value of $K$ under the usual order-selection condition that the minimal expected risk is attained at some finite $K^{\star}$ and does not improve for larger $K$. Establishing analogous consistency statements for the BIC and ICL in the present SALMoE setting would require separate regularity and identifiability arguments for these criteria, and is beyond the scope of this paper.

\bibliography{sn-bibliography}
\end{document}